\documentclass[prx,nofootinbib,amsmath,amssymb, notitlepage, twocolumn,superscriptaddress,longbibliography]{revtex4-1}
\usepackage[pdftex,colorlinks=true,linkcolor=darkblue,citecolor=blue,urlcolor=darkred]{hyperref}
\usepackage{braket}
\usepackage{amsmath,amsfonts, amssymb, amsthm, dsfont}
\usepackage{yfonts,bbm}
\usepackage{bm,breqn}
\usepackage{mathrsfs}
\usepackage{graphicx}
\usepackage[caption=false]{subfig}
\usepackage[normalem]{ulem}
\usepackage{verbatim}
\usepackage{xcolor}

\usepackage{float}

\theoremstyle{definition}
\newtheorem{definition}{Definition}
\newtheorem{theorem}{Theorem}

\definecolor{darkblue}{rgb}{0.,0.,0.4}
\definecolor{darkred}{rgb}{0.5,0.,0.}

\usepackage{verbatim}
\usepackage{stackengine}

\usepackage{tikz}
\usetikzlibrary{calc}
\usepackage{multirow}
\usepackage[capitalise]{cleveref}

\renewcommand{\tilde}{\widetilde}

\newcommand{\bs}{\textbf{s}}
\renewcommand{\bm}{\textbf{m}}
\newcommand{\be}{\textbf{e}}
\newcommand{\bl}{\textbf{l}}
\renewcommand{\mod}{~\textrm{mod}~}

\newcommand{\E}{\mathbb{E}}
\newcommand{\tr}{\text{tr}}
\renewcommand{\Pr}{\mathrm{Pr}}

\newcommand{\identity}{\mathbb{I}}

\newcommand{\hilbert}{\mathcal{H}}
\newcommand{\supp}{\mathsf{supp}}
\newcommand{\argmin}{\mathrm{argmin}}
\newcommand{\argmax}{\mathrm{argmax}}
\newcommand{\polylog}{{\rm polylog}}

\newcommand{\poly}{\mathrm{poly}}
\newcommand{\EPR}{\mathrm{EPR}}
\newcommand{\GHZ}{\mathrm{GHZ}}

\newcommand{\eqfig}[2]{\vcenter{\hbox{\includegraphics[height=#1]{#2}}}}

\begin{document}

\title{Mixed-state Quantum Phases: Renormalization and Quantum Error Correction}

\author{Shengqi Sang}
\affiliation{\PI}
\affiliation{\UW}
\affiliation{\KITP}

\author{Yijian Zou}
\affiliation{\PI}
\affiliation{\ST}

\author{Timothy H. Hsieh}
\affiliation{\PI}

\newcommand*{\PI}{Perimeter Institute for Theoretical Physics, Waterloo, Ontario N2L 2Y5, Canada}
\newcommand*{\UW}{Department of Physics and Astronomy, University of Waterloo, Waterloo, Ontario N2L 3G1, Canada}
\newcommand*{\KITP}{Kavli Institute for Theoretical Physics, University of California, Santa Barbara, CA 93106, USA}
\newcommand*{\ST}{Stanford Institute for Theoretical Physics, Stanford CA 94305, USA}

\begin{abstract}
Open system quantum dynamics can generate a variety of long-range entangled mixed states, yet it has been unclear in what sense they constitute phases of matter. 
To establish that two mixed states are in the same phase, as defined by their two-way connectivity via local quantum channels, we use the renormalization group (RG) and decoders of quantum error correcting codes. We introduce a real-space RG scheme for mixed states based on local channels which ideally preserve correlations with the complementary system, and we prove this is equivalent to the reversibility of the channel's action. 
As an application, we demonstrate an exact RG flow of finite temperature toric code in two dimensions to infinite temperature, thus proving it is in the trivial phase. 
In contrast, for toric code subject to local dephasing, we establish a mixed state toric code phase using local channels obtained by truncating an RG-type decoder and the minimum weight perfect matching decoder. 
We also discover a precise relation between mixed state phase and decodability, by proving that local noise acting on toric code cannot destroy logical information without bringing the state out of the toric code phase.

\end{abstract}
\maketitle


\section{Introduction}

Understanding quantum phases of matter is a central task of quantum many-body physics. The traditional focus is on pure states, which are typically ground states of local Hamiltonians. However, in many physical contexts ranging from finite temperature systems to open system dynamics \cite{diss2,diss1}, one is required to deal with mixed states. Recently, in the context of non-equilibrium quantum simulators and computers, there has been significant progress in constructing many different examples of nontrivial mixed states from the effect of local decoherence on symmetry-protected topological and long-range entangled pure states~\cite{de2022symmetry, ma2023average, fan2023diagnostics, bao2023mixed, lee2023quantum, zou2023channeling, ma2023topological, wang2023intrinsic, chen2023separability} or from protocols involving measurement and feedback~\cite{lu2023mixed, zhu2022nishimoris, lee2022decoding, li2023decodable}. 

Given the increasing wealth of examples, it is thus desirable to have a general framework of mixed-state phases and in particular a notion of renormalization for distilling universal long-range properties of a phase. Furthermore, in the class of mixed states obtained by decohering an error-correcting code, there are remarkable instances \cite{fan2023diagnostics} in which mixed state entanglement measures undergo a transition at the same point at which encoded information is lost.  This motivates us to understand the precise connection between mixed state phase transitions and error correction thresholds. In this work we will address these questions.

One way of defining {\it pure-state} phases is via local unitary (LU) circuits~\cite{chen2010LU}: two states are in the same phase if there is a short-depth LU circuit that connects them. This is based on the physical intuition that phases should be defined by long-range properties and representatives only differ in their local properties. 
For mixed states, an analogous definition was proposed by Coser and Perez-Garcia~\cite{coser2019classification}: two mixed states $\rho_1$ and $\rho_2$ are in the same phase if there exists a pair of short-time evolution with local Lindbladians from $\rho_1$ to $\rho_2$ and from $\rho_2$ to $\rho_1$. 
The major difference with the LU definition of pure state phases is that two-way connections are needed since channels are not in general reversible. Another difference is that unlike pure states of interest, there is often no notion of Hamiltonian, gap, or adiabatic path to furnish the local transformations required to connect two mixed states (see however Ref.~\cite{rakovszky2023defining} for recent developments).  These make establishing the existence of a mixed state phase much more challenging.


We draw inspiration from the real-space renormalization group (RG), which has played a major role in statistical mechanics and quantum many-body physics. The idea dates back to Wilson \cite{wilson1975renormalization} and Kadanoff \cite{Kadanoff_1966} who proposed that under block spin transformations, statistical mechanical systems flow to fixed points whose properties are easier to characterize. In the context of quantum many-body systems, real-space RG has led to the development of powerful numerical algorithms, including density matrix renormalization group (DMRG) \cite{white1992dmrg}, multiscale entanglement renormalization ansatz (MERA) \cite{vidal2007entanglement}, as well as theoretical tools, including matrix product states (MPS) \cite{Fannes_1992}, projected entangled pair states (PEPS) \cite{verstraete2004renormalization}, etc.
However, thus far, real-space RG has predominantly been applied to coarse-grain {\it pure} quantum states.


In this work, we define a real-space RG scheme for mixed states involving local channel (LC) transformations to establish the existence of mixed-state phases. We define an ``ideal'' RG to consist of local channels acting on blocks which preserve correlations between different blocks, 
and we prove that the actions of such correlation-preserving channels can be reversed by another channel, thus establishing the phase equivalence of the fine-grained and coarse-grained states. 
As an example, we construct an ideal RG for the two-dimensional toric code at finite temperature and show that the temperature monotonically increases under coarse graining and thus the state does not possess topological order. 

We also consider mixed states obtained by applying local decoherence to quantum error correction codes. There is a notion that the logical information in topological codes is protected by long-range entanglement. With a definition of mixed-state phases, we can make its relation to error correction precise. We prove that short-range correlated noise (represented by a local quantum channel) cannot destroy logical information without also transitioning out of the mixed state topologically ordered phase. 
We illustrate these connections in the example of toric code subject to local dephasing noise, for which we demonstrate the existence of a toric code mixed state phase by constructing (1) a real-space RG scheme based on the Harrington decoder \cite{harrington2004analysis} and (2) (quasi-)local channels based on truncating the minimum weight perfect matching (MWPM) algorithm. Our local version of MWPM can potentially be used for efficiently detecting the toric code mixed-state phase in experiments \cite{doi:10.1126/science.abi8794, doi:10.1126/science.abi8378}.

We mention that several prior related works~\cite{swingle2016mixed, lin2021entanglement} developed a mixed state RG scheme based on purification of the mixed state, which is generally different from our scheme but in some cases can furnish the local channels required in our scheme. Refs.~\cite{de2022matrix, cirac2017matrix} defined an RG fixed point condition for one-dimensional matrix product density operators and related them to boundaries of two-dimensional topological order. Refs.~\cite{lake2022exact, cong2022enhancing} demonstrated how quantum convolutional neural networks~\cite{cong2019quantum} can furnish RG schemes for detecting non-trivial pure state phases. 

This paper is structured as follows. In Sec.~\ref{sec: phase} we define LC transformations and mixed state phase equivalence. 
In Sec.~\ref{sec: mixed RG} we formulate the real-space RG for mixed-states and discuss its implications. 
In Secs.~\ref{sec: examples} to~\ref{sec: noisy_tc} we analyze several examples, including the dephased GHZ state, thermal toric code, and dephased toric code. In Sec.~\ref{sec: the_thm} we prove a relation between decodability and mixed-state phases.

\begin{figure*}
\centering
\includegraphics[width=.99\linewidth]{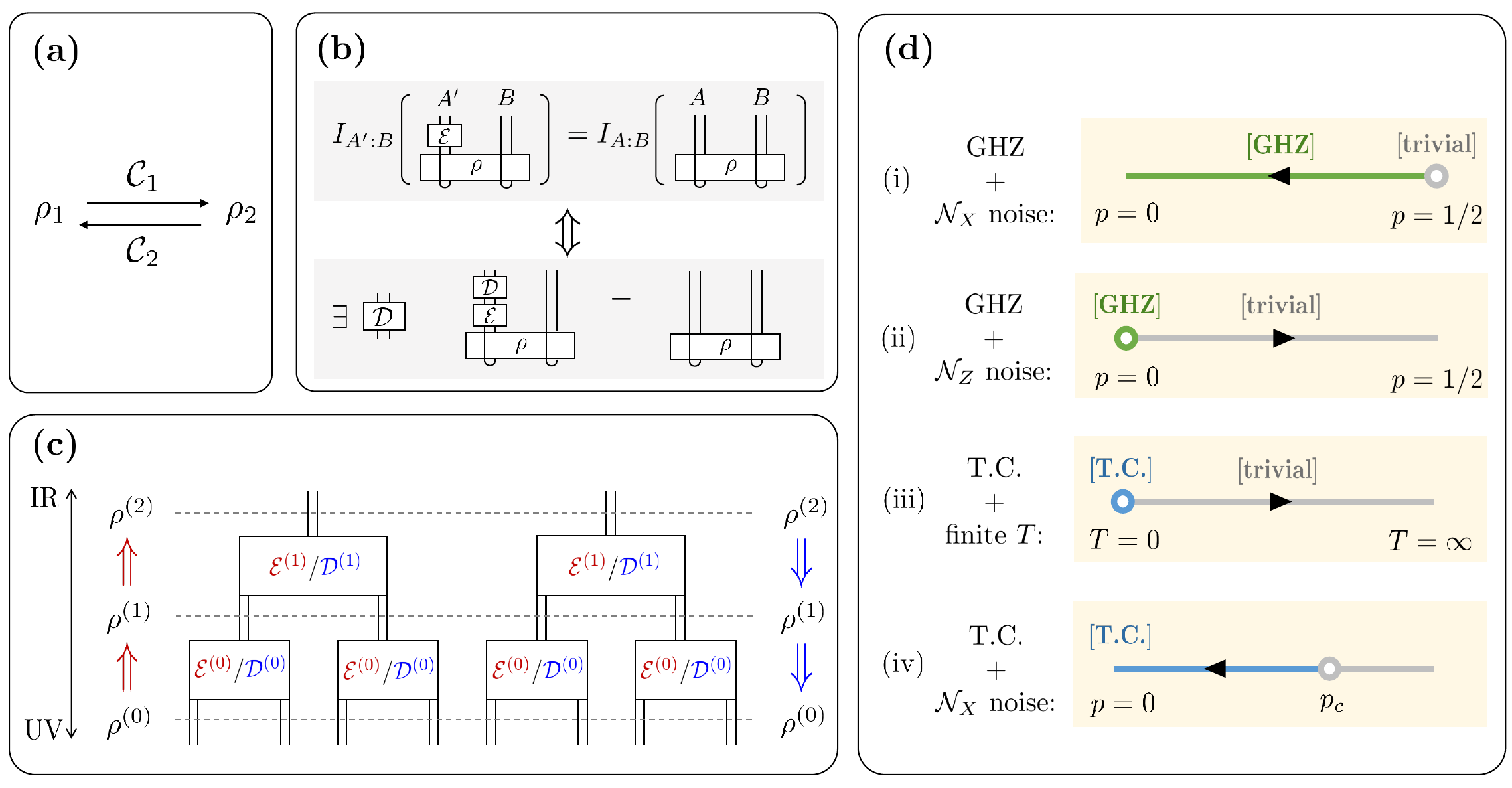}
\caption{
\textbf{(a)} Definition of mixed-state phase equivalence adopted in this work. Two many-body mixed states $\rho_1$ and $\rho_2$ are in the same phase if there is a pair of low-depth spatially local quantum channels $\mathcal{C}_1$ and $\mathcal{C}_2$ such that $\rho_2\approx\mathcal{C}_1(\rho_1)$ and $\rho_1\approx\mathcal{C}_2(\rho_2)$.~~
\textbf{(b)} Illustration of the correlation-preserving criterion in Def.\ref{def: correlation-preserving}. For a given bipartite mixed state $\rho_{AB}$, a quantum channel $\mathcal{E}$ acting on one party is correlation-preserving if it leaves the mutual information between two parties invariant. Thm.\ref{thm: thm2} shows that $\mathcal{E}$ is correlation-preserving if and only if its action can be reversed by another channel $\mathcal{D}$. ~~
\textbf{(c)} Mixed-state RG consists of local channels ($\mathcal{E}$s) which coarse-grain degrees of freedom within a block. 
After iterating, all short-range correlations of the input state are discarded and only long-range ones remain. 
If all coarse-graining channels satisfy the correlation-preserving criterion, then the whole RG process can be reversed, by running from top to bottom and replacing each $\mathcal{E}$ with its recovery map $\mathcal{D}$. ~~
\textbf{(d)} Phase diagrams and RG flows of 4 exemplary mixed states studied in the Sec.\ref{sec: examples}. All 4 states come from perturbing a long-range entangled pure state in an incoherent way: In examples (i, ii) the pure state is the GHZ state, and in (iii, iv) it is the toric code state. In examples (i, ii, iv) the incoherent perturbation is a dephasing noise with strength $p$ acted upon the state, while in (iii) the perturbation is a non-zero temperature. The mixed-state phase corresponding to the GHZ state and the toric code state are denoted by [GHZ] and [T.C.], respectively. 
}
\label{fig: main_fig}
\end{figure*}

\section{Local channel transformations and definition of mixed-state phases}\label{sec: phase}

\subsection{Local channel transformations}
We define local channel transformations following the proposal in~\cite{hastings2011nonzero}.
\begin{definition}[Local channel (LC) transformation]  \label{def: lc_transformation}
On a given lattice of linear dimension $L$, a range-$r$ LC transformation is a quantum channel composed of the following steps:\\
(1) Adding qubits to each lattice site, all initialized in the $\ket{0}$ state;\\
(2) Applying a range-$r$ unitary circuit $U$ on the lattice;\\
(3) Tracing out some qubits on each lattice site. 
\end{definition}

The range of a circuit is defined as the maximal range of each unitary gate times the depth of the circuit. 
Henceforth if $r$ is not specified for an LC transformation, it is assumed that $r/L\rightarrow 0$ in the thermodynamic limit. 

The major difference between local channel transformations and local unitary ones (LU)~\cite{chen2010LU} is step (3). In local unitary transformations, a qubit can be discarded only when it is disentangled from the rest of the system.
In that context, (1) and (3) are inverse operations, and hence LU transformations are invertible. 
In contrast, LC transformations allow discarding a qubit that is still entangled with the rest of the system, \textit{i.e.} $\rho_{i, \bar i}\neq \rho_{i} \otimes \rho_{\bar i}$ with $i$ being the qubit to be discarded and $\bar i$ being the rest of the system. 
As a result, LC transformations are generically non-invertible. 

\begin{figure}[h]
    \centering
    \includegraphics[width=0.9\linewidth]{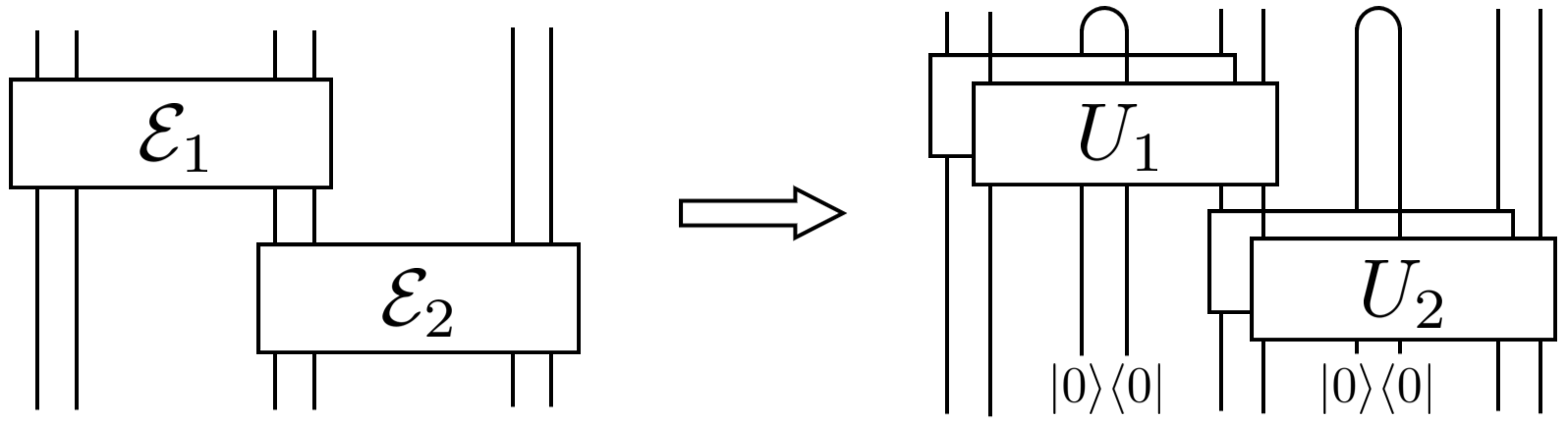}
    \caption{A circuit of local channel gates represented as an LC transformation.}
    \label{fig: circuit_rewritting}
\end{figure}
LC transformations constitute a broad class of operations including any circuit composed of local channel gates, \textit{i.e.} channels that only act on local domains of sites. 
To show this, one needs the Stinespring dilation theorem: any quantum channel $\mathcal{E}_{X\rightarrow Y}$ can be rewritten as: 
\begin{equation}
    \mathcal{E}(\cdot) = \tr_{A'}\left(U ((\cdot)\otimes\ket{0}\bra{0}_{A}) U^\dagger\right)
\end{equation}
where $U$ is a unitary map from $X\cup A$ to $Y\cup A'$. In other words, any quantum channel can be implemented by adding some degrees of freedom, applying a unitary on the joint system, and discarding some degrees of freedom. 
Applying the theorem to a circuit of channel gates, one can first replace each channel gate with its Stinesping dilation form. Then one can move forward all the ancillae addition to the beginning of the circuit and postpone all the tracing-out to  the end of the circuit. We graphically illustrate this in Fig.\ref{fig: circuit_rewritting}. Furthermore, any finite-time  local Lindbladian evolution can also be approximated by LC transformations by trotterizing the continuous dynamics.

\subsection{Definition of mixed-state phase equivalence}

When studying (pure) ground states of gapped local Hamiltonians, two many-body states are defined to be in the same phase if one can be turned into the other through a LU transformation \cite{chen2010LU}. The definition reflects the idea that phases of matter should be characterized by long-range properties of the state, and should remain unchanged under reversible local modifications. 


LC transformations, albeit local, are generally not reversible and can destroy long-range correlations. 
As an example, if one starts from an arbitrary state $\ket{\psi}$ and applies the amplitude damping channel $\mathcal{E}_{\rm damping}(\cdot) := \tr(\cdot) \ket{0}\bra{0}$ to each qubit in the system, the resulting state would be a product state $\ket{0}^{\otimes L}$ without any non-trivial long-range correlation. 
On the other hand, an LC transformation's ability to create correlations is no stronger than LU ones. 
This follows from the fact that an LC transformation is some LU followed by discarding some degrees of freedom. 

Thus the connectivity under LC transformations induces a partial order relation among mixed-states. States are ordered according to the amount of long-range correlation they possess:
if $\rho_2 = \mathcal{C}(\rho_1)$ for some LC transformation $C$, then $\rho_1$ has at least as much long-range correlation as $\rho_2$. 
This naturally leads to the following definition of mixed-state phase equivalence \footnote{The definition resembles the one taken in \cite{coser2019classification}, where a pair of LC transformations is replaced by a pair of (quasi-)local Lindbladian evolutions. A similar definition also appears in \cite{ma2023average} when defining mixed-state symmetry-protected topological orders.}

\begin{definition}[Mixed-state phase equivalence] \label{def: phase_equivalence}
On a given lattice, two many-body mixed states $\rho_1$ and $\rho_2$ are in the same phase if there exists a pair of LC transformations $C_1$ and $C_2$ such that $C_1(\rho_1)\approx \rho_2$ and $C_2(\rho_2)\approx \rho_1$. 
\end{definition}
Several clarifications regarding the definition:
\begin{itemize}
    \item \textbf{Mixed-states of interest}: Though the definition above does not assume any restrictions on states $\rho_{1,2}$, we are interested in  physically relevant mixed states such as local Hamiltonian Gibbs states at finite temperature, gapped ground states subject to decoherence, and steady states of local Lindbladians. 
    \item \textbf{The precise meaning of `$\approx$'}: This requires some distance measure of mixed states. For instance, we could define two mixed states $\rho \approx \sigma$ if and only if $F(\sigma,\rho)>1-\epsilon$ for some small $\epsilon>0$, where $F(\sigma, \rho):=||\sqrt{\sigma}\sqrt{\rho}||_1$ is the (Uhlmann) fidelity.
    \item \textbf{Ranges of LC transformations $C_{1,2}$}: In general we only require the range to be much smaller than the linear size of the lattice $\rho_{1,2}$ is defined on. But as we will see later, it will be sufficient to have a range $r=O(\polylog(L/\epsilon))$ when $\rho_{1,2}$ have finite correlation length.
\end{itemize}

The definition is a natural generalization of pure-state phase equivalence defined through LU transformations. When restricting to pure many-body states, one can show that two states $\ket{\psi_1}$ and $\ket{\psi_2}$ are of the same mixed-state phase if and only if $\ket{\psi_1}$ and $\ket{\psi_2}\otimes \ket{\phi}$ are of the same pure-state phase for some invertible state $\ket{\phi}$. We provide a proof in App.\ref{app: mixed_and_pure}. 

Product states, \textit{e.g.} $\ket{0}^{\otimes L}$, are states without any long-range correlations.
This is also reflected by the partial order relation under LC circuits: any state can be turned into the product state by a LC transformation which only consists of the amplitude damping channel.
Thus we identify the \textit{trivial phase} as the set of states that can be LC transformed from the product state. In other words, a mixed state is in the trivial phase if it can be written as $\rho_{\rm trivial} = \mathcal{C}\left[\ket{0}^{\otimes L}\bra{0}^{\otimes L}\right]$
for some LC transformation $\mathcal{C}$. This is equivalent to requiring that the state can be locally purified into a short-range entangled pure state.

We comment that the above definition treats quantum and classical correlations on the same footing. 
As an example, the state $\rho = \frac{1}{2} (\ket{0^{\otimes L}}\bra{0^{\otimes L}} + \ket{1^{\otimes L}}\bra{1^{\otimes L}})$
is a classical ensemble of $L$ spins which is non-trivial under the above definition, because it has classical long-range correlation. To single out states that contain long-range classical correlation only, we can define a state to be in a \textit{classical phase} if it can be written as $
    \rho_{\rm classical} = \mathcal{C}(\rho_{\Pr(\bs))})$
for some LC transformation $\mathcal{C}$. Here $\rho_{\Pr(s)}:=\sum_{\bs}\Pr(\bs)\ket{\bs}\bra{\bs}$ is a classical distribution $\Pr(\bs)$ of product states $\{\ket{\bs}:\bs\in\{0,1\}^L\}$ represented as a density matrix.

\section{Real-space RG of quantum mixed-states} \label{sec: mixed RG}
To answer whether two given states $\rho_1$ and $\rho_2$ are in the same phase, we need to either construct a pair of local channel transformations or prove their nonexistence. 

Recall that when studying pure-state phases of ground states, adiabatic paths between Hamiltonians provide a convenient way of obtaining phase equivalence. 
Let $\ket{\psi_1}$ and $\ket{\psi_2}$ be ground-states of local Hamiltonians $H_1$ and $H_2$. If there is a path from $H_1$ to $H_2$ in the space of local Hamiltonians such that the energy gap remains $O(1)$ throughout the path, there is a standard way to construct a LU transformation connecting $\ket{\psi_1}$ to $\ket{\psi_2}$ which establishes the phase equivalence \cite{hastings2005quasiadiabatic}. 
For mixed states, there is generally no counterpart to adiabatic paths. 

In this section, we introduce the mixed-state real-space renormalization group (RG) as an alternative way to find LC connections and identify mixed-state phases. 

\subsection{From pure-state RG to mixed-state RG}\label{sec: pure_to_mixed_rg}

Conceptually, RG transformation in classical and quantum statistical mechanics is an iterative coarse-graining process that discards short-range degrees of freedom while preserving long-range ones. 
The idea of using real-space renormalization to study zero-temperature physics of lattice quantum systems (`numerical RG') was pioneered by Wilson when considering impurity problems \cite{wilson1975renormalization}, and was later generalized and developed into a series of powerful RG-based numerical methods including DMRG \cite{white1992dmrg}, entanglement RG \cite{vidal2007entanglement}, etc. 
We refer to all of them as pure-state RGs, in contrast to the mixed-state RG we introduce in this work.
\footnote{Another aspect of pure-state RG is that it preserves the low-energy physics of the system. But here we emphasize this less because for mixed-states, there may not be a notion of energy or Hamiltonian.}

For the sake of presentation, we restrict our attention to one-dimensional systems and only focus on tree-like RG circuits (see Fig.~\ref{fig: rg_circuits}). All the main ideas can be easily generalized to more sophisticated RG circuit structures, \textit{e.g.} entanglement RG circuits \cite{vidal2007entanglement}, as well as higher dimensional systems.

\begin{figure}
    \centering
    \includegraphics[width=1\linewidth]{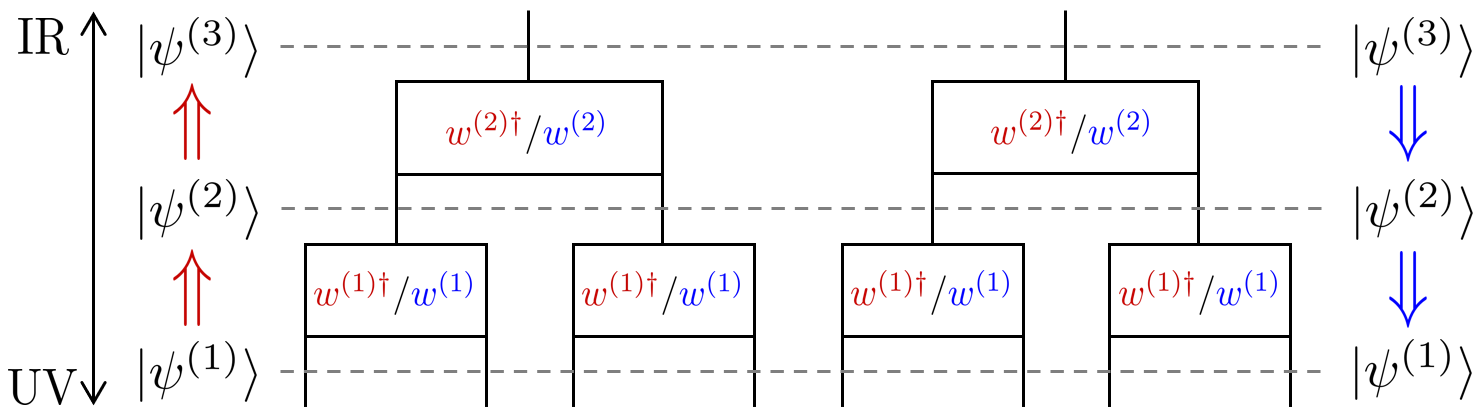}
    \caption{
    \textbf{Real-space RG transformation of pure states}-- 
    Circuit representation of two iterations of pure state RG transformation. At the $\ell$-th iteration, the coarse-graining isometry $w^{(\ell)}$ is determined by the level's input state $\ket{\psi^{(\ell)}}$ using Eq.\eqref{eq: white_rule}. By applying the circuit  from bottom to top (red arrows), all the short-range features of the initial UV state are gradually discarded, and only long-range ones are kept in the IR state $\rho^{(\ell\rightarrow\infty)}$. By  applying the circuit from top to bottom (blue arrows), the circuit generates the UV state $\ket{\psi^{(1)}}$.}
    \label{fig: rg_circuits}
\end{figure}

Pure-state RG, in its simplest form, involves partitioning the lattice into consecutive blocks each with size $b$ and applying a coarse-graining map $w^{\dagger}_B$ to each block $B$ of a pure state $\ket{\psi}$.  
More specfically, coarse-graining involves truncating the Hilbert space, and $w_B$ is an isometry satsifying $w^{\dagger}_B w_B = \mathbb{I}$. 
As proposed \cite{white1992dmrg} for the density matrix renormalization group (DMRG) algorithm, the optimal choice of $w_B$ that preserves all correlations between $B$ and its complement is given by
\begin{equation}\label{eq: white_rule}
      \supp~ w_B w^\dagger_B =\supp~\rho_{B}  
\end{equation}
where $\rho_{B} := \tr_{\bar B} (\ket{\psi}\bra{\psi})$ is the reduced density matrix of the block $B$. $\supp~ K$ of a positive semi-definite matrix $K$ means the subspace spanned by $K$'s eigenstates with positive eigenvalues. 
If the original state $\ket{\psi}$ has area law entanglement $S_B \equiv -\tr[\rho_B \log \rho_B] = O(1)$, then each block is efficiently coarse-grained into a constant dimensional Hilbert space independent of the original block size \footnote{Rigorously speaking, it is only proven that the ground state of a gapped local Hamiltonian satisying the area law can be represented as an MPS with a bond dimension that grows sublinearly with the system size \cite{arad2013area}. However, in pratice, it is usually true that a finite bond dimension suffices to reproduce an accurate wavefunction for arbitrary large (even infinite) system sizes. }. 



Now we turn to 1D mixed-states. In contrast to the pure state case, physical mixed states (\textit{e.g.} ones mentioned below Def.\ref{def: phase_equivalence}) typically have volume-law scaling of $S_B$, leading to inefficient compression using the $w_B$ selected according to Eq.\eqref{eq: white_rule}. This is because $S_B$ results from not only correlations between $B$ and the complementary system $\bar{B}$ but also between $B$ and a purifying environment $E$ of the mixed state.  The latter is the non-universal information that should be discarded.  We thus need a new criterion for finding the coarse-graining map. 

To motivate the criterion we introduce, we observe that Eq.\eqref{eq: white_rule} can be interpreted as the solution to the optimization problem: 
\begin{equation} \label{eq: pure_optimization}
\begin{aligned}
    &\argmin_{w_B} \, \dim_{\rm out}(w^\dagger_B)\\
    s.t.~~~
    &I_{B:\bar B}(w_B^\dagger \ket{\psi}) = I_{B:\bar B}(\ket{\psi}),
\end{aligned}
\end{equation}
where $\dim_{\rm out}(w^\dagger_B)$ is the output dimension of $w^\dagger_B$ and $I_{X:Y}:=S_X+S_Y-S_{XY}$ is the quantum mutual information, a measure of correlations between two parties $X$ and $Y$. The constraint has clear physical meaning in the context of RG: by preserving $I_{B:\bar B}$, it preserves all the long-distance correlation within $\ket{\psi}$. 
We thus use `$I_{B:\bar B}$ preserving' condition as a guideline to generalize Eq.\eqref{eq: white_rule} to mixed-states.

\subsection{Correlation-preserving map}

We make the argument above more precise:
\begin{definition}[correlation-preserving maps]\label{def: correlation-preserving}
For a given bipartite quantum state $\rho=\rho_{A B}$, a quantum channel $\mathcal{E}_{A\rightarrow A'}$ acting on $A$ is \textit{correlation-preserving} with respect to $\rho_{AB}$ if it satisfies
\begin{equation*}
    I_{A':B}(\mathcal{E}_{A\rightarrow A'}(\rho)) = I_{A:B}(\rho)
\end{equation*}
\end{definition}

It is worth noting that a channel being correlation-preserving or not depends on both the input state and the bipartition: the same map $\mathcal{E}$ that is correlation-preserving with respect to one $(\rho, B)$ pair may not be so with respect to another pair. 

Recalling that the motivation for defining an RG scheme is to establish equivalence between two mixed states by finding a local channel transformation and its inverse, ideally we would like $\mathcal{E}$'s action on $\rho$ to be reversible.  Conveniently, the two desired properties (correlation-preserving and reversibility) are equivalent, as we prove in the following theorem. 

\begin{theorem}\label{thm:equivalence}
For a given bipartite quantum state $\rho=\rho_{A B}$, the map $\mathcal{E}_{A\rightarrow A'}$ is correlation-preserving if and only if there exists another quantum channel $\mathcal{D}_{A'\rightarrow A}$, such that:
\begin{equation*}
    \rho = \mathcal{D}_{A'\rightarrow A}\circ \mathcal{E}_{A\rightarrow A'}(\rho)
\end{equation*}
\end{theorem}
\begin{proof}
\noindent (reversibility $\Rightarrow$ correlation-preserving)
According to the quantum data processing inequality, a channel acting only on $A$ cannot increase correlations between $A$ and $B$:
\begin{equation}
    I_{A:B}(\rho)
    \geq
    I_{A':B}(\mathcal{E}(\rho))
    \geq
    I_{A:B}(\mathcal{D}\circ\mathcal{E}(\rho))
    = 
    I_{A:B}(\rho)
\end{equation}
Thus $I_{A:B}(\rho) = I_{A':B}(\mathcal{E}(\rho))$.

\noindent (correlation-preserving $\Rightarrow$ reversibility)
Let $W$ be an isometry from $A$ to $A'\cup E$ that dilates the channel $\mathcal{E}_{A\rightarrow A'}$:
\begin{equation}
    \mathcal{E}_{A\rightarrow A'}(\cdot) := \tr_{E}\left[W(\cdot)W^\dagger\right],
\end{equation}
where $E$ is an ancillary system, and let $\sigma_{A'EB}=W \rho W^\dagger$. Then we have the following relation:
\begin{equation}
    I_{A:B}(\rho)=I_{A'E:B}(\sigma_{A'EB})=I_{A':B}(\sigma_{A'EB}).
\end{equation}
The second equality, due to the correlation-preserving property, implies that $I_{B:E|A'}(\sigma_{A'EB})=0$ and $B-A'-E$ forms a quantum Markov chain. 
Thus there is a channel $\mathcal{T}_{A'\rightarrow A'E}$ that reconstructs $\sigma_{A'EB}$ from $\mathcal{E}_{A\rightarrow A'}(\rho)=\sigma_{A'B}=\tr_E\sigma_{A'EB}$ alone:
\begin{equation}
    \mathcal{T}_{A'\rightarrow A'E}(\sigma_{A'B}) = \sigma_{A'EB}.
\end{equation}
The map $\mathcal{T}_{A'\rightarrow A'E}$ is the Petz recovery map \cite{petz1988sufficiency}: 
\begin{equation}
    \mathcal{T}_{A'\rightarrow A'E}(\cdot) := \sigma_{A'E}^{1/2}\left( \sigma^{-1/2}_{A'}(\cdot)\sigma^{-1/2}_{A'}\otimes \identity_E \right)\sigma_{A'E}^{1/2}
\end{equation}
We can then choose the inverse channel $\mathcal{D}$ to be 
\begin{equation}\label{eq:recovery_map}
    \mathcal{D}_{A'\rightarrow A}(\cdot) = \tr_R\left(U_{W}^\dagger\mathcal{T}_{A'\rightarrow A'E}(\cdot)U_W\right),
\end{equation}
where $U_W:~A\cup R\rightarrow A'\cup E$ is a unitary operator that `completes' the isometry $W:~A\rightarrow A'\cup E$, namely:
\begin{equation}
    W(\cdot)W^\dagger = U_W\left((\cdot)\otimes\ket{0}_R\bra{0}\right)U^\dagger_W.
\end{equation}
\end{proof}

We remark that the relation between correlation-preserving and reversibility is robust in one direction. More precisely, if the channel $\mathcal{E}_{A\rightarrow A'}$ almost preserves correlation
\begin{equation}\label{eq: approx_correlation_preserving}
    I_{A:B}(\rho)-I_{A':B}(\mathcal{E}_{A\rightarrow A'}(\rho))=\epsilon, 
\end{equation}
then there exists an almost perfect recovery channel $\mathcal{D}_{A'\rightarrow A}$ such that
\begin{equation}\label{eq: approx_correlation_preserving2}
    F(\rho, \mathcal{D}\circ\mathcal{E}(\rho))\geq 2^{-\epsilon/2}.
\end{equation}
A proof, based on approximate quantum Markov chains \cite{fawzi2015quantum}, can be found in App.\ref{app: approx_correlation_preserving}. 
The robustness property is desirable especially when we would like to numerically search for the correlation-preserving channel $\mathcal{E}$.

When using $\mathcal{E}$ for the purpose of coarse-graining, the target space of the channel should be as small as possible.
This corresponds to solving the following optimization problem:
\begin{equation} \label{eq: mixed_optimization}
\begin{aligned}
    &\argmin_{\mathcal{E}_{A\rightarrow A'}} \, \dim\hilbert_{A'}\\
    s.t.~~~
    &I_{A:\bar{A}}(\rho) -I_{A':\bar{A}}(\mathcal{E}(\rho)) \leq \epsilon
\end{aligned}
\end{equation}
with $\epsilon$ taken to be a small number or zero. The problem is analogous to Eq.\eqref{eq: pure_optimization} for pure states, which has Eq.\eqref{eq: white_rule} as an explicit solution. The current problem, in contrast, has no known explicit solution.
In fact, the problem is closely related to the mixed-state quantum data compression problem, which is under active exploration in quantum information theory. We refer interested readers to Refs.~\cite{anshu2021incompressibility, khanian2022general, khanian2022strong, kato2023exact} for recent discussions on the problem.

To search for a good coarse-graining map for a given state, one can either numerically solve the optimization problem Eq.\eqref{eq: mixed_optimization} (in this case the robustness property is crucial for the purpose of estimating error), or try to construct the channel analytically by exploiting the special structure of the given state, as we do later when studying examples in Sec.\ref{sec: examples}. 

We point out that two familiar coarse-graining schemes in 1D, one for quantum ground states and one for classical statistical mechanics models, are in fact correlation-preserving maps. The first one is the Hilbert space truncation  reviewed in Sec.\ref{sec: pure_to_mixed_rg} using the rule Eq.\eqref{eq: white_rule}. Since this scheme preserves the entropy of a block, it satisfies the correlation-preserving condition (for a pure state $\ket{\psi_{AB}}$, $S_A = \frac{1}{2}I_{A:B}$). The other example is Kadanoff's block spin decimation of classical spin chains. 
Consider the Gibbs state of a classical spin chain with nearest-neighbor interaction, but written as a quantum mixed-state:
\begin{equation}
    \rho_\beta \propto \sum_{\textbf{s}=s_0...s_L}\exp\left(-\beta\sum_i h_i(s_i, s_{i+1})\right)\ket{\textbf{s}}\bra{\textbf{s}}
\end{equation}
The state is classical because it is diagonal in the computational basis $\ket{\bs}=\ket{s_1...s_L}$. 
For each block $B=\{i_1, ..., i_b\}$, the block spin decimation corresponds to a quantum channel that traces out all spins in $B$ other than $i_1$. 
This operation is correlation-preserving with respect to $B' := B\cup \{i_{b+1}\}$, because:
\begin{equation}
    I_{B':\overline{ B'}}(\rho) = 
    I_{\{i_1, i_{b}+1\}: \overline{ B'}}(\rho)
\end{equation}
which is a consequence of the Markov property of the Gibbs distribution.

\subsection{Ideal mixed-state RG}\label{sec: iRG}
In this section, we formulate an ideal real-space RG scheme built from local correlation-preserving channels.

Assume $\rho$ is a many-body mixed-state on a lattice with linear size $L$. We further assume that we have constructed, either numerically or analytically, a series of coarse-graining transformations $\{\mathcal{C}^{(0)}, \mathcal{C}^{(1)},...,\mathcal{C}^{(\ell)},...\}$ acting on $\rho$ sequentially. In 1D, each $\mathcal{C}^{(\ell)}$ may have one of the structures shown in Fig.~\ref{fig: rg_circuits}, or any other structure as long as it is composed of at most $O(1)$-layer of local channels. 
This leads to an `RG flow' of mixed-states:
\begin{equation}
\rho=\rho^{(0)}\xrightarrow{\mathcal{C}^{(0)}}\rho^{(1)}\xrightarrow{\mathcal{C}^{(1)}}...\xrightarrow{\mathcal{C}^{(\ell-1)}}\rho^{(\ell)}\xrightarrow{\mathcal{C}^{(\ell)}}...
\end{equation}
along which the level of coarse-graining increases gradually. 

Each state $\rho^{(\ell)}$ is supported on a coarse-grained lattice $\mathcal{L}^{(\ell)}$ with an $L^{(\ell)}:=L/b^{\ell}$ linear size. The chain has a length at most $\sim \log_b L$, after which the state is supported on $O(1)$ number of sites.

We call this RG process \textit{ideal} if every channel gate $\mathcal{E}$ within each $\mathcal{C}^{(\ell)}$ is correlation-preserving with respect to its input and the prescribed bipartition.

As a direct consequence of Thm.\ref{thm:equivalence}, ideal RG is reversible. More specifically, there exists a series of local `fine-graining' transformations $\{\mathcal{F}^{(0)}, \mathcal{F}^{(1)},...,\mathcal{F}^{(\ell)},...\}$ that recovers the original mixed-state from its coarse-grained version by gradually adding local details:
\begin{equation}
\rho^{(0)}\xleftarrow{\mathcal{F}^{(0)}}\rho^{(1)}\xleftarrow{\mathcal{F}^{(1)}}...\xleftarrow{\mathcal{F}^{(\ell-1)}}\rho^{(\ell)}\xleftarrow{\mathcal{F}^{(\ell)}}...
\end{equation}
where each $\mathcal{F}^{(\ell)}$ is the `reversed' channel of $\mathcal{C}^{(\ell)}$, obtained by replacing each channel $\mathcal{E}$ within $\mathcal{F}^{(\ell)}$ by its corresponding recovery map $\mathcal{D}$ (see Thm.\ref{thm:equivalence}). In graphical notation, if 
\begin{equation}
        \mathcal{C}^{(\ell)} = \left(\eqfig{.8cm}{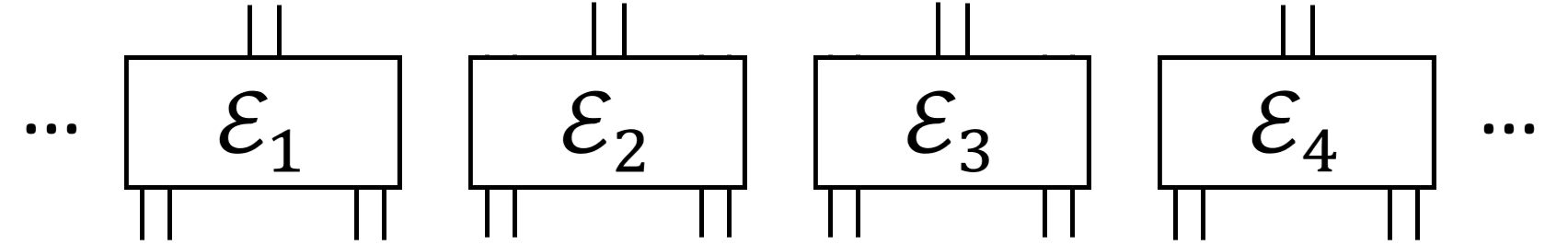}\right)
\end{equation}
then
\begin{equation}
    \mathcal{F}^{(\ell)} := \left(\eqfig{.8cm}{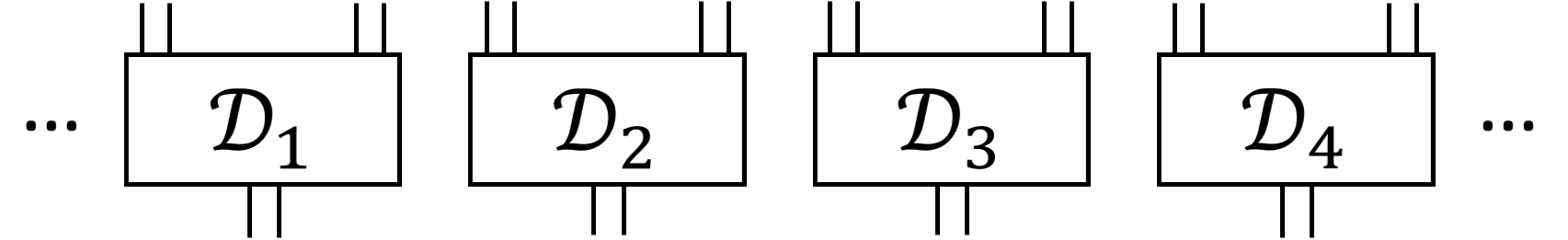}\right)
\end{equation}
In both plots, the state ($\rho^{(\ell)}$ for $\mathcal{C}^{(\ell)}$ and $\rho^{(\ell+1)}$ for $\mathcal{F}^{(\ell)}$), is inserted from the bottom.

Coarse-graining and fine-graining maps also establish relations between operators at different coarse-graining levels. Let $O^{(\ell)}$ be any operator (not necessarily local) defined on the lattice $\mathcal{L}^{(\ell)}$. Then the following relation holds:
\begin{equation}
\begin{aligned}
    \tr\left(\rho^{(\ell)} O^{(\ell)}\right)
    &=
    \tr\left(\mathcal{F}^{(\ell)}\circ\mathcal{C}^{(\ell)}(\rho^{(\ell)}) O^{(\ell)}\right)\\
    &=
    \tr\left(\rho^{(\ell+1)}\mathcal{F}^{\dagger (\ell)}( O^{(\ell)})\right),
\end{aligned}
\end{equation}
where $\mathcal{F}^{\dagger}$ is $\mathcal{F}$'s dual map, defined through the relation $\tr(A\cdot\mathcal{F}(B))\equiv\tr(\mathcal{F}^{\dagger}(A)\cdot B)$. 
Thus we can define the coarse-grained operator of $O^{(\ell)}$ through:
\begin{equation}
    O^{(\ell+1)} := \mathcal{F}^{\dagger (\ell)}( O^{(\ell)})
\end{equation}
so that: 
\begin{equation}
    \langle{O^{(\ell)}}\rangle^{(\ell)} = \langle{O^{(\ell+1)}}\rangle^{(\ell+1)}
\end{equation}
where $\langle{~\cdot~}\rangle^{(\ell)} := \tr(\rho^{(\ell)}(~\cdot~))$.

It is worth noting that $\mathcal{F}^\dagger=\mathcal{F}^{(\ell)\dagger}$ does not preserve operator multiplication: if $O^{(\ell)} = o_1^{(\ell)} o_2^{(\ell)}$ then it is possible that $\mathcal{F}^\dagger(O^{(\ell)})=O^{(\ell+1)}\neq o_1^{(\ell+1)} o_2^{(\ell+1)}=\mathcal{F}^\dagger(o_1^{(\ell)})\mathcal{F}^\dagger(o_2^{(\ell)})$. 
However, if two operators $o_1$ and $o_2$ are spatially well-separated, then the multiplication is preserved:
\begin{equation}\label{eq: operator_product}
    \mathcal{F}^{\dagger}(o_1^{(\ell)} o_2^{(\ell)}) = \mathcal{F}^{\dagger}(o_1^{(\ell)}) \mathcal{F}^{\dagger}(o_2^{(\ell)})
\end{equation}
Here `well-separated' means that the lightcones for $o_1$ and $o_2$, as determined by the circuit structure of $\mathcal{C}^{(\ell)}$ (or equivalently that of $\mathcal{F}^{(\ell)\dagger}$), are non-overlapping. We illustrate the definition of the lightcone and a proof of the above equation in App.\ref{app: operator_product}. 
This guarantees that long-distance behavior of all the $k$-point functions are preserved along an ideal RG:
\begin{equation}
    \langle o_{1}^{(\ell)}o_{2}^{(\ell)}...o_{k}^{(\ell)}\rangle^{(\ell) }
    =
    \langle o_{1}^{(\ell+1)}o_{2}^{(\ell+1)}...o_{k}^{(\ell+1)}\rangle^{(\ell+1)}
\end{equation}
where $\{o_i\}$ are mutually well-separated local operators.

\subsection{From mixed-state RG to mixed-state quantum phases}\label{sec: rg_to_phase}

In this section we discuss how to use RGs, both ideal and non-ideal ones, to furnish the two-way LC transformations required to establish the phase equivalence of two mixed states (Def.\ref{def: phase_equivalence}). 

Assume that for the state of interest $\rho$, we have found a (not necessarily ideal) real-space RG process $\{\mathcal{C}^{(1)}, \mathcal{C}^{(2)},...,\mathcal{C}^{(\ell)},...\}$. We further assume that the RG has a well-defined fixed-point state $\rho^{(\infty)}$, whose mixed-state phase of matter is presumably easy to identify.

Intuitively, this RG can be treated as an LC transformation connecting $\rho$ to $\rho^{(\infty)}$ \footnote{When discussing LC transformations in this section, we fix the reference lattice as the one that $\rho=\rho^{(0)}$ is defined upon. The renormalized state $\rho^{(\ell)}$ can be considered as supported on a sub-lattice with $L/b^{\ell}$ sites.}. 
To rigorously show this according to Defs.\ref{def: lc_transformation} and \ref{def: phase_equivalence}, one needs to show that the sequence $\{\rho^{(\ell)}\}$ converges toward $\rho^{(\infty)}$ fast enough. 
More concretely, we need to show that there exists $\ell^*$ such that:
\begin{enumerate}
    \item The channel $\mathcal{C}^{(\ell^*)}\circ...\circ\mathcal{C}^{(2)}\circ\mathcal{C}^{(1)}$ is an LC transformation. Noticing that $\mathcal{C}^{(\ell)}$ entails $b^{\ell}$ range of non-locality while an LC transformation can have at most $o(L)$ range, the condition is equivalent to requiring $\ell^{*} \lesssim \log L^\alpha$, for some $\alpha<1$. 
    \item The state after $\ell^*$ iterations is close enough to $\rho^{(\infty)}$, namely $F(\rho^{(\ell^*)}, \rho^{(\infty)})>1-\epsilon$ for a small $\epsilon$.
\end{enumerate}

We find that such an $\ell^*$ does exist in many cases when the fixed-point state $\rho^{(\infty)}$ has a finite correlation length.
More specifically, in such cases, the fidelity function satisfies the form:
\begin{equation}\label{eq: convergence_condition1}
    F(\rho^{(\ell)}, \rho^{(\infty)}) \simeq \exp(-\alpha\ \theta^{(\ell)} L^{(\ell)})
\end{equation}
for some $\alpha=O(1)$ and a positive coefficient $\theta^{(\ell)}$.  Further, $\theta^{(\ell)}$ displays a power-law iteration relation under each coarse-graining step:
\begin{equation}\label{eq: convergence_condition2}
     \theta^{(\ell+1)} \lesssim (\theta^{(\ell)}) ^ {\gamma} ~~~{\rm when}~~~ \theta^{(\ell)}\rightarrow 0_+
\end{equation}
for some coefficient $\gamma>1$.  

As detailed in the App.\ref{app: convergence}, Eqs.\eqref{eq: convergence_condition1},\eqref{eq: convergence_condition2} guarantee that choosing
\begin{equation}
    \ell^* \sim \log\log (L/\epsilon)
\end{equation}
is sufficient to have $F(\rho^{(\ell)}, \rho^{(\infty)})>1-\epsilon$. 
We remark that one can let $\epsilon$ be as small as $(\poly L)^{-1}$ but still guarantee that $\ell^*$ steps of RG is a $(\polylog L)$ -range LC transformation.

In the App.\ref{app: convergence} we show that conditions Eq.\eqref{eq: convergence_condition1} and Eq.\eqref{eq: convergence_condition2} hold for: 
\begin{itemize}
    \item 1D pure-state RG of a matrix product state
    \item Gibbs state of a classical statistical mechanics model flowing toward a non-critical fixed point
    \item All examples we study in Sec.\ref{sec: examples}
\end{itemize}

So far in this section, we have shown that RG can be viewed as an LC transformation connecting $\rho$ to $\rho^{(\infty)}$.
Recalling that the phase equivalence is defined through two-way LC connections, we have to find another LC channel connecting $\rho^{(\infty)}$ to $\rho$ to conclude that the two states are in the same phase.

If the RG is an ideal one, it is composed of correlation-preserving channels and thus reversible. In this case, the other direction comes from the `reversed' RG process  ${\rm RG}^{-1}=\{\mathcal{F}^{(1)}, \mathcal{F}^{(2)},...\mathcal{F}^{(\ell)},...\}$ which we discussed in Sec.\ref{sec: iRG}.  Similar to the forward RG, there is the issue of convergence concerning whether ${\rm RG}^{-1}$ can be treated as an LC transformation. But the discussion is completely parallel to the one for the forward RG. 
Thus in this case, the LC bi-connection is established as
\begin{equation}
    \rho 
    ~\xrightarrow{{\rm RG}}~
    \rho^{(\infty)} 
    ~\xrightarrow{{\rm RG}^{-1}}~
    \rho
\end{equation}
and we can conclude $\rho$ and $\rho^{(\infty)}$ are in the same phase. 

Next we discuss what we can learn from a non-ideal RG. 
One class of mixed states of significant interest is a long-range entangled pure-state $\ket{\psi}$ subject to local decoherence represented by an LC transformation, and an important question is whether or not the decohered state in the same phase as $\ket{\psi}$. 
In this setting, one direction of the connection is already given by the decoherence. Therefore, if the RG (ideal or not) has $\ket{\psi}$ as the fixed point:
\begin{equation}
    \ket{\psi} \xrightarrow{\rm decohere} \rho \xrightarrow{\rm RG} \rho^{(\infty)}= \ket{\psi}\bra{\psi},
\end{equation}
then an LC bi-connection is established and $\rho$ and $\ket{\psi}$ are in the same phase. 
But on the other hand, if the fixed-point is not in the same phase as $\ket{\psi}$, then we cannot determine $\rho$'s phase because no bi-connection is identified.

\section{Overview of examples} \label{sec: examples}
In the remaining sections we use our formalism to understand the quantum phases of several many-body mixed-states of recent interest. 

In all the examples, the mixed-state is obtained by `perturbing' a long-range entangled pure state, either through incoherent noise or finite temperature. The question we address is whether the states before and after the perturbation are in the same phase.

The long-range entangled pure state is chosen to be either the Greenberger–Horne–Zeilinger (GHZ) state or Kitaev's toric code state. 
In most examples, the LC circuits for identifying phases take the form of RG. The coarse-graining maps therein are either constructed according to the correlation-preserving criterion (Def.\ref{def: correlation-preserving}), or inspired by decoders of quantum error correcting codes.

In the App.\ref{app: spt_rg} we include an example of a mixed symmetry-protected topological (SPT) state and its associated mixed-state RG.

\section{Noisy GHZ states}\label{sec: noisy_ghz}
The many-body GHZ state, defined as
\begin{equation}
    \ket{\GHZ_L} := \frac{1}{\sqrt{2}}\left( \ket{0^{\otimes L}} +\ket{1^{\otimes L}}\right),
\end{equation}
has long-range entanglement, \textit{i.e.} it can not be generated from a product state using any one-dimensional LU (or LC) transformation from a product state. 
 In this section, we study the effect of dephasing noise on this state. 

For convenience of analysis, we let $L = b^{\ell_{\rm max}}$ for some integer $\ell_{\rm max}$ and an odd integer $b$. The state can be rewritten as:
\begin{equation}\label{eq: ghz_tn_rep}
    \ket{\GHZ_L} = w_b^{\otimes b^{\ell_{\rm max}-1}} \cdot w_b^{\otimes b^{\ell_{\rm max}-2}} ... ~w_b^{\otimes b^1} \ket{+} 
\end{equation}
where $w_b=\ket{0^{\otimes b}}\bra{0} + \ket{1^{\otimes b}}\bra{1}$ is an isometry and $\ket{+}=\frac{1}{\sqrt{2}}(\ket{0}+\ket{1})$. This provides a tree tensor network representation of the state (see Fig.\ref{fig:ghz_tn}), as well as a way of blocking sites when performing RG.

\begin{figure}[h]
\centering
\includegraphics[width=.4\linewidth]{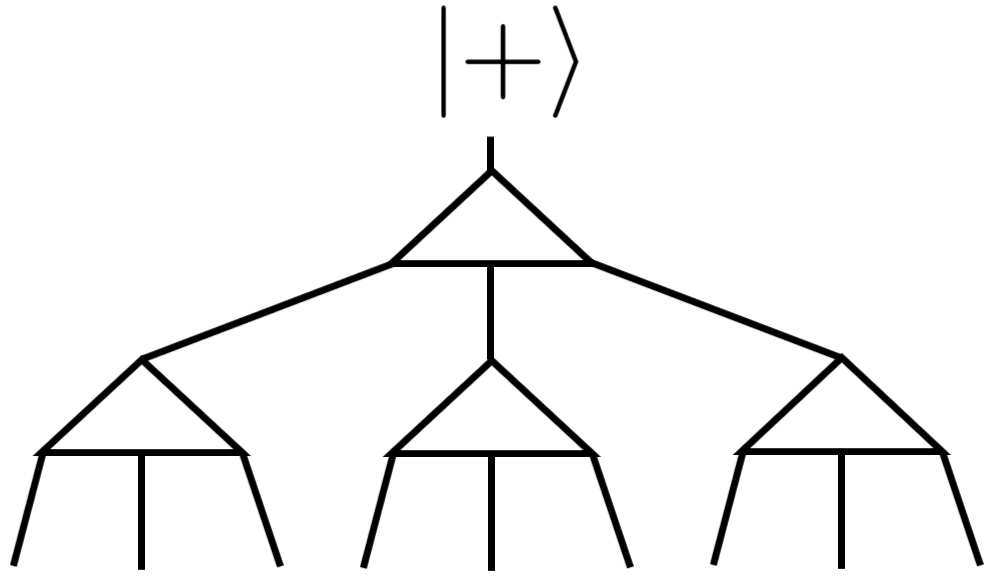}
\caption{Tree tensor network of a GHZ state with $L=b^\ell=9$, $b=3$, $\ell=2$. Each triangle represents an isometry $w$ (Eq.\eqref{eq: ghz_tn_rep}). By replacing the state at the top with a generic single qubit state $\ket{\psi}$, the same tensor network encodes $\ket{\psi}$ into a codeword state of the quantum repetition code.}
\label{fig:ghz_tn}
\end{figure}

We consider a setting in which each qubit experiences the same noise, as modeled by a single qubit channel $\mathcal{N}$, resulting in the mixed state
\begin{equation} \label{eq: ghz_noisy_state_general}
\rho_{L} := \mathcal{N}^{\otimes L}(\ket{\GHZ_L}\bra{\GHZ_L})\\
\end{equation}

We remark that the GHZ state is closely related to the quantum repetition code, whose codespace is spanned by $\ket{0^{\otimes L}}$ and $\ket{1^{\otimes L}}$. 
It is known that quantum information stored in a quantum repetition code is robust against bit-flip noise ($X$ dephasing noise), but not phase-flip noise ($Z$ dephasing noise). As we will see in this section, the robustness of the GHZ state's long-range entanglement has a parallel behavior when it is subjected to these two types of noise. 
We postpone a detailed discussion of the relation between mixed-state phases and quantum coding properties to Sec.\ref{sec: noisy_tc}. 

\subsection{Bit-flip noise}
We first consider dephasing each qubit in the $X$ direction:
\begin{equation}\label{eq: X_dephase}
    \mathcal{N}(\cdot) = \mathcal{N}^X_p(\cdot) := (1-p)(\cdot)+p X(\cdot) X,
\end{equation}
in which each qubit is flipped with probability $p$.

The resulting state is 
\begin{equation}
    \rho^X_{p, L} 
    = \frac{1}{2}\sum_{\bs} p^{|\bs|}(1-p)^{L-|\bs|}
    (
    \ket{\bs}\bra{\bs} + 
    \ket{\bs}\bra{\bar \bs} + 
    \ket{\bar\bs}\bra{\bs} + 
    \ket{\bar\bs}\bra{\bar\bs}
    )
\end{equation}
where $|\bs|:=\sum_i s_i$ is the number of $1$ in the bitstring $\bs$, and $\bar \bs$ is the bitwise complement of $\bs$ . Since $\rho^{X}_p = \rho^{X}_{1-p}$, we only consider $p\in(0,0.5]$.

Inspired by decoders for the quantum repetition code, we use the $b$-qubit \textit{majority-vote} channel as the coarse-graining map. 
To define it, we first introduce the unitary operator that re-parametrizes the bitstring:
\begin{equation}
    U\ket{\bs} := \ket{\mathsf{maj}(\bs)}\otimes\ket{\mathsf{diff}(\bs)}
\end{equation}
where $\mathsf{maj}(\bs)$ takes the majority vote of the bits within $\bs$:
\begin{equation}
    \mathsf{maj}(\bs) :=
    \left\{
    \begin{array}{ll}
    0     &  {\rm if}~~ |\bs|<b/2\\
    1     &  {\rm if}~~ |\bs|\geq b/2
    \end{array}
    \right.
\end{equation}
and $\mathsf{diff}(\bs)$ is a length $(b-1)$ bitstring that records pairwise difference of $\bs$:
\begin{equation}
    \mathsf{diff}(\bs)_i := (s_{i+1}-s_i)\mod 2~~~i=1,2,...,b-1
\end{equation}
Then the majority vote channel can be written as:
\begin{equation} \label{eq: majority_vote}
    \mathcal{E}_b(\cdot) := \tr_{2}(U(\cdot)U^\dagger)
\end{equation}
where $\tr_2$ denotes tracing out the pairwise difference information, regarded as unimportant short-distance degrees of freedom in the current example. 

We inspect the state's RG flow under $\mathcal{E}_b$, the coarse-graining map:
\begin{equation}
\begin{aligned}
    &\mathcal{E}_b^{\otimes {L/b}}(\rho^X_{p, L})\\
    =& \mathcal{E}_b^{\otimes {L/b}} \circ (\mathcal{N}^X_p)^{\otimes L} (\ket{\GHZ_L}\bra{\GHZ_L})\\
    =& \left(\mathcal{E}_b \circ (\mathcal{N}^X_{p})^{\otimes b} \circ \mathcal{U}_{w_b}\right)^{\otimes L/b}\left(\ket{\GHZ_{L/b}}\bra{\GHZ_{L/b}}\right)
\end{aligned}
\end{equation}
where $\mathcal{U}_{w_b}(\cdot):=w_b(\cdot)w_b^\dagger$ and we applied the relation $\ket{\GHZ_L}=w_b^{\otimes L/b}\ket{\GHZ_{L/b}}$. 
Thus after one iteration of coarse-graining, the resulting state is a GHZ state with $1/b$ of the original size subject to a `renormalized' noise channel, which is still $X$-dephasing  (see App.\ref{app: ghz_iteration} for a derivation):
\begin{equation}\label{eq: ghz_dephase_channel_iteration}
    \mathcal{E}_b \circ (\mathcal{N}_{p}^X)^{\otimes b} \circ \mathcal{U}_{w_b} = \mathcal{N}_{p'}^X,
\end{equation}
but with a renormalized noise strength $p' = \sum_{k=(b+1)/2}^{b} \binom{b}{k}p^{k}(1-p)^{b-k}$ .


Thus we obtain an exact description of the state's RG flow:
\begin{equation}
    \rho^{(\ell)} = (\mathcal{N}^X_{p^{(\ell)}})^{\otimes L^{(\ell)}}(\ket{\GHZ_{L^{(\ell)}}} \bra{\GHZ_{L^{(\ell)}}})
\end{equation}
where $L^{(\ell)}=L b^{-\ell}$ is the renormalized system size at the $\ell$-th iteration, and $ p^{(\ell+1)} = \sum_{k=(b+1)/2}^{b} \binom{b}{k}(p^{(\ell)})^{k}(1-p^{(\ell)})^{b-k}$. It is straightforward to check that $p=0$ and $p=1/2$ are the two fixed points of the RG transformation. 

Around $p=1/2$, the iteration relation has the asymptotic behavior:
\begin{equation}
    (p'-1/2) \simeq g(b) (p-1/2)
\end{equation}
where $g(b):=2^{-(b-1)}\sum_{k=(b+1)/2}^b (2k-b)\cdot\binom{b}{k}>1$. Thus this is an unstable fixed point.
Exactly at $p=1/2$, the fixed-point state is:
\begin{equation}
    \rho^X_{1/2,~L} = \frac{1}{2^L}\sum_{\bs\in\{0,1\}^L} \left(\ket{\bs}\bra{\bs} + \ket{\bs}\bra{\bar \bs}\right).
\end{equation}
The state is better understood in the eigenbasis of Pauli $X$ operators, \textit{i.e.} $\{\ket{0^X}=\frac{1}{\sqrt 2}(\ket{0}+\ket{1}),$~$\ket{1^X}=\frac{1}{\sqrt 2}(\ket{0}-\ket{1})\}$:
\begin{equation}
    \rho^X_{1/2,~L} = \frac{1}{2^{L-1}}\sum_{\bs\in\{0,1\}^L} \ket{\bs^X}\bra{\bs^X} ~\delta(|\bs|=0\mod 2)
\end{equation}
In this basis, the state is a uniform distribution of bitstrings with even parity. 
Since it is diagonal in this basis, the state is a classical state (recall the definition in Sec.\ref{sec: phase}) and is not in the same phase as the GHZ state \footnote{In fact, the state $\rho^{X}_{1/2,L}$ is in the same phase as the product state. One may prepare the state using LC by first preparing a one dimensional cluster state of $2L$ spins and then tracing out qubits on all odd sites.}.

On the other hand, $p=0$ is a stable fixed-point attracting $p\in [0,0.5)$. 
Starting from any state in the interval, the RG process gradually removes entropy within the state and brings it back to the noiseless state at $p=0$, \textit{i.e.} $\ket{\GHZ}$. 
We thus obtain the following LC transformation bi-connection:
\begin{equation} \label{eq: ghz_X_biconnection}
    \ket{\GHZ} 
    ~\xrightarrow{\mathcal{N}_X}~
    \rho^X_p 
    ~\xrightarrow{\rm RG}~ 
    \ket{\GHZ} \quad\quad p\in[0,0.5)
\end{equation}
Thus we can conclude $\rho_p^X$ and $\ket{\GHZ}$ are in the same phase. The analysis shows that the $X-$dephasing noise acts as an irrelevant perturbation with respect to the $\ket{\GHZ}$ state and its long-range entanglement.

Besides establishing the phase equivalence, the bi-connection in Eq.\eqref{eq: ghz_X_biconnection} also yields information on the entanglement structure of the dephased state $\rho_p^X$. 
Consider two sufficiently large subregions of the system, referred to as $A$ and $B$, for which one can always choose the RG blocking scheme such that coarse-graining channels never act jointly on $A$ and $B$. Let $E_{A:B}(\cdot)$ be \textit{any} quantum or classical correlation measure between $A$ and $B$ that satisfies the data processing inequality. Then due to the bi-connection we have:
\begin{equation}
\begin{aligned}
    & E_{A:B}(\ket{\GHZ}) \geq E_{A:B}(\rho_p^X) \geq E_{A:B}(\ket{\GHZ})\\
    \Rightarrow~
    & E_{A:B}(\rho_p^X) = E_{A:B}(\ket{\GHZ}).
\end{aligned}
\end{equation}

Some examples of correlation measures are quantum mutual information, entanglement negativity, and entanglement of formation \& distillation. 
All quantities are easy to compute analytically for the GHZ state but are difficult to obtain for the mixed-state $\rho_p^X$ by other means. 

We point out that all conclusions in this section hold also for the dephasing noise along other directions in the $X$-$Y$ plane. This can be most easily seen by noticing that the expression Eq.\eqref{eq: ghz_dephase_channel_iteration} holds for any dephasing direction in the $X$-$Y$ plane. 
As we will see in the next subsection, the $Z$-dephasing acts very differently.

\subsection{Phase-flip noise}
Next we consider the GHZ state under another type of noise, namely the phase-flip or Z-dephasing:
\begin{equation}
    \mathcal{N}^Z_p(\cdot) := (1-p)(\cdot)+p Z(\cdot) Z
\end{equation}
which leads to the density matrix
\begin{equation}
\begin{aligned}
    \rho^Z_{p, L} 
    = 
    &\frac{1}{2}[
    \ket{0^{\otimes L}}\bra{0^{\otimes L}}+
    \ket{1^{\otimes L}}\bra{1^{\otimes L}}+\\
    &
    (1-2p)^L(\ket{0^{\otimes L}}\bra{1^{\otimes L}}+
    \ket{1^{\otimes L}}\bra{0^{\otimes L}})
    ]
\end{aligned}
\end{equation}
In the thermodynamic limit, the off-diagonal term vanishes for any $p\notin\{0,1\}$ and the state converges to a classical state $ \frac{1}{2} (\ket{0^{\otimes L}}\bra{0^{\otimes L}}+\ket{1^{\otimes L}}\bra{1^{\otimes L}})$. This already indicates that the state is in a different phase from the GHZ state. 

To construct the RG for this mixed state, we still use the majority-vote channel $\mathcal{E}_b$ (Eq.\eqref{eq: majority_vote}) as the coarse-graining map.  An important difference of this case compared to the bit-flip case is that the majority vote channel is now correlation-preserving with respect to $\rho^Z_{p,L}$. 
To see this, we first verify the following relations:
\begin{equation}
    \begin{aligned}
        \mathcal{E}_b \circ \mathcal{U}_{w_b} &= \mathcal{I}\\
        (\mathcal{N}_p^Z)^{\otimes b} \circ \mathcal{U}_{w_b} &= \mathcal{U}_{w_b} \circ \mathcal{N}_{p'}^Z
    \end{aligned}
\end{equation}
where $\mathcal{I}$ is the identity channel and $p'$ is given later in Eq.\eqref{eq: ghz_z_p_iteration}. These equations imply that:
\begin{equation}
    \mathcal{U}_{w_b} \circ \mathcal{E}_{b}\left(\rho_{p, L}^Z\right) = \rho_{p, L}^Z
\end{equation}
where $\mathcal{U}_{w_b} \circ \mathcal{E}_{b}$ is applied to any block of $b$ sites. Thus $\mathcal{E}_b$ is reversible and correlation-preserving with respect to $\rho_{p, L}^Z$.

Following a similar calculation as in the bit-flip noise case, we obtain that the state after one step of RG maintains the same form:
\begin{equation}
    \mathcal{E}_b^{\otimes L/b} (\rho^{Z}_{p,L}) = \rho^{Z}_{p',L/b}
\end{equation}
but with a renormalized noise strength
\begin{equation} \label{eq: ghz_z_p_iteration}
    p' = \frac{1}{2}(1 - (1-2p)^b),
\end{equation}

The iteration relation has $p=0$ as an unstable fixed point, around which 
    $p' \simeq bp$, and also $p = 1/2$ as a stable fixed point, around which $(p'-1/2) = (p-1/2)^b$.

Since the RG is ideal, it leads to the following LC  bi-connection:
\begin{equation}
    \rho^Z_{1/2} 
    ~\xrightarrow{{\rm RG}^{-1}}~
    \rho^Z_p 
    ~\xrightarrow{\rm RG}~
    \rho^Z_{1/2} \quad\quad p\in[0,0.5)
\end{equation}
and the analysis shows the noisy state is of the same phase as the classical state $\frac{1}{2}(\ket{0^{\otimes L}}\bra{0^{\otimes L}}+\ket{1^{\otimes L}}\bra{1^{\otimes L}})$.  Therefore, for the GHZ state the phase flip noise is relevant and destroys the long-range entanglement therein with an arbitrarily small strength.

\section{Thermal toric code state}\label{sec: thermal_tc_total}
In this section and the next, we discuss two mixed-states related to $\mathbb{Z}_2$ topological order. In Sec.\ref{sec: tc_review} we review key properties of the toric code model and define the notations. In Sec.\ref{sec: thermal_tc} we construct an ideal RG to explicitly show that any finite temperature Gibbs state of the toric code is in the trivial phase.

\subsection{Review of the toric code model}\label{sec: tc_review}
We consider a square lattice with periodic boundary conditions and qubits on the links.
Kitaev's toric code model has the Hamiltonian
\begin{equation} \label{eq: tc_hamiltonian}
    H = -\sum_{\square\in P} A_\square - \sum_{+\in V} B_+
\end{equation}
where $A_\square = \prod_{i\in\square} X_i$ and $B_+ =\prod_{i\in +} Z_i$. $P, V$ represent plaquettes and vertices, respectively.

Since all terms in the Hamiltonian commute with each other, their common eigenstates can be used to label the Hilbert space. But in order to construct a complete basis, we need two more operators $\tilde{X}_{1,2}=\prod_{i\in S_{1,2}} X_i$, where $S_{1}, S_2$ are the two homotopically inequivalent non-contractable loops on the torus. 
Each $\tilde X_i$ commutes with $A_\square$s and $B_+$s, thus all of them together define a basis for the Hilbert space:
\begin{equation}\label{eq: anyon_basis}
\begin{aligned}
    \ket{\bm=m_1...m_{|P|}; \be=e_1...e_{|V|}; \bl=l_1 l_2},\quad m_i,e_i,l_i\in\{0,1\}
\end{aligned}
\end{equation}
satisfying
\begin{equation}
    \begin{aligned}
        A_{\square_i}\ket{\bm; \be; \bl} &= (-1)^{m_i}\ket{\bm; \be; \bl}\\
        B_{+_i}\ket{\bm; \be; \bl} &= (-1)^{e_i}\ket{\bm; \be; \bl}\\
        \tilde X_i\ket{\bm; \be; \bl} &= (-1)^{l_i}\ket{\bm; \be; \bl}
    \end{aligned}
\end{equation}
We call this the anyon number basis in contrast to the computational basis. If $m_i=1$, there is a plaquette anyon (or $m$ anyon) at the corresponding plaquette; while if $e_i=1$, there is a vertex anyon (or $e$ anyon) at the corresponding vertex. The operator identities $\prod_\square A_\square = 1$ and $\prod_+ B_+ = 1$ enforce that the total number of either type of anyon must be even:
\begin{equation}
        \pi(\bm) = 0 \quad \pi(\be) = 0,
\end{equation}
where the function $\pi(\cdot)$ evaluates the total parity of a bit string, \textit{i.e.} $\pi(\bs) := (\sum_i s_i \mod 2)$.

The Hamiltonian's 4-dimensional ground state subspace $V$ is spanned by anyon-free states:
\begin{equation}\label{eq: tc_V}
    V := \mathsf{span}\{ ~\ket{\bm=\textbf{0}; ~\be=\textbf{0}; ~\bl}:~~ \bl\in\{00,01,10,11\} \}
\end{equation}
States within this subspace are locally indistinguishable, \textit{i.e.} $\rho_{A}=\tr_{\bar A}(\ket{\psi}\bra{\psi})$ is independent of $\ket{\psi}\in V$ whenever $A$ is a topologically trivial region. 


We define a mixed-state $\rho$ to be in the toric code phase if it is LC bi-connected to states within $V$, namely:
\begin{equation}
    \rho_a 
    ~\xrightarrow{\mathcal{C}_1}~
    \rho 
    ~\xrightarrow{\mathcal{C}_2}~
    \rho_b
\end{equation}
for some states $\rho_a, \rho_b$ within $V$, and some LC transformations $\mathcal{C}_1$, $\mathcal{C}_2$. 



\subsection{RG of the thermal toric code state} \label{sec: thermal_tc}
We consider the Gibbs state of the toric code model Eq.\eqref{eq: tc_hamiltonian} $\rho_\beta \propto \exp(-\beta H)$ at inverse temperature $\beta$.
Ref.~\cite{hastings2011nonzero} showed that this state for finite $\beta$ is not long-range entangled, and here we reproduce the conclusion by constructing an ideal mixed-state RG under which the state flows to a trivial one.

We notice that the density matrix $\rho_\beta$ is diagonal in the anyon number basis (Eq.~\eqref{eq: anyon_basis}):
\begin{equation}
    \rho_\beta\ket{\bm; \be; \bl} \propto \ket{\bm; \be; \bl},
\end{equation}
and is thus a classical mixture of different anyon configurations, with probabilities
\begin{equation}\label{eq: thermal_anyon_dist}
\begin{aligned}
    \Pr(\bm,\be,\bl) 
    :=& \bra{\bm; \be; \bl}\rho_\beta\ket{\bm;\be;\bl} \\
    =& \Pr_m(\bm)\Pr_e(\be)\Pr_l(\bl)
\end{aligned}
\end{equation}
in which the three types of degrees of freedom are independent: 
\begin{equation}\label{eq: thermal_tc_anyon_dist_separated}
\begin{aligned}
    \Pr_m(\bm) &= C_\beta ~\delta \left(\pi(\bm)=0\right)  \prod_i p_\beta^{m_i}(1-p_\beta)^{1-m_i}\\
    \Pr_e(\be) &= C_\beta ~\delta \left(\pi(\be)=0\right)  \prod_i p_\beta^{e_i}(1-p_\beta)^{1-e_i}\\
    \Pr_l(\bl) &= 1/4
\end{aligned}
\end{equation}
with $p_\beta=\frac{e^{-\beta}}{e^\beta + e^{-\beta}}$ and $C_\beta$  a normalization constant. 

The key property is that $m$ anyons on each plaquette (and $e$ anyons on vertices) are independently excited with probability $p_\beta$, up to a global constraint that the total number of each anyon type is even.
This allows us to find an ideal RG, as we need only preserve the local anyon parity $\pi(\bm_B), \pi(\be_B)$ of a block $B$ to maintain correlations between the block and its complement. 


We now describe how to coarse-grain to preserve this parity information.  Consider the following quantum channel acting on $12$ qubits in a $2\times 2$ block of plaquettes:
\begin{equation}\label{eq: theram_tc_E_X}
    \mathcal{E}^X(\cdot) := \sum_{\bm\in\{0,1\}^{\otimes 4}}U_{\bm}P_{\bm}(\cdot)P_{\bm}U_{\bm}^\dagger
\end{equation}
where $P_{\bm}$ is the projector to the subspace with anyon configuration $\bm$, and the unitary operator $U_\bm$ is a product of Pauli $Z$ matrices that brings $\ket{\bm=m_1m_2m_3m_4}$ to $\ket{\pi(\bm) 0 0 0 }$. 
For instance, if we assume plaquettes are labeled as ${\begin{smallmatrix}1&2\\3&4\end{smallmatrix}}$ and $\bm=0110$, then $U_{\bm}$ can be $Z_{12}Z_{13}$, where $Z_{12(13)}$ is the Pauli-Z matrix acting on the qubit separating 1 and 2 (1 and 3). We remark that $U_{\bm}$ only acts on the inner four qubits. 

In other words, $\mathcal{E}^X$ first measures the anyon configuration within the block and then applies a unitary gate depending on the measurement outcome that pushes all anyons to the top-left plaquette. Since $m$-anyon is its own anti-particle, the top-left plaquette ends up with $\pi(\bm)$ anyons while the other three end up with $0$. Importantly, neither step disturbs the distribution of $e$-anyons.

$\mathcal{E}^X$ is a correlation-preserving map with respect to the state $\rho_\beta$, and one can explicitly check that its action on $\rho_\beta$ can be reversed by the following channel:
\begin{equation}
    \mathcal{D}^X(\cdot) := \sum_{\bm} \Pr(\bm|\pi(\bm)) U^\dagger_\bm P^1_{\pi(\bm)}(\cdot)P^1_{\pi(\bm)}U_\bm
\end{equation}
where $P^1_{x}$ is the projector to the subspace with $m_1=x$. 
The action of $\mathcal{D}^X$ can be intuitively understood as follows: It first measures the anyon occupancy of the site $1$, which we recall is the only site that may host an anyon after the action of $\mathcal{E}^X$. Then based on the measurement outcome (referred to as $x$), it randomly generates an anyon configuration on the block according to the distribution $\Pr(\bm|\pi(\bm)=x)$.

Analogously, there is a channel for each $2\times 2$ block of vertices (\textit{i.e.} a block of plaquettes of the dual lattice) that coarse-grains $e$-anyons:
\begin{equation}\label{eq: theram_tc_E_Z}
    \mathcal{E}^Z(\cdot) := \sum_{\be\in\{0,1\}^{\otimes 4}}U_{\be}P_{\be}(\cdot)P_{\be}U_{\be}^\dagger
\end{equation}
where $P_\be$ are projector for $e$-anyon configurations and $U_\be$ a product of $X$ operators that brings $\ket{\be=e_1e_2e_3e_4}$ to $\ket{\pi(\be) 0 0 0 }$. $\mathcal{E}^Z$ only moves $e$-anyons and commutes with $\mathcal{E}^X$. 

\begin{figure}
    \centering
    \includegraphics[width=.95\linewidth]{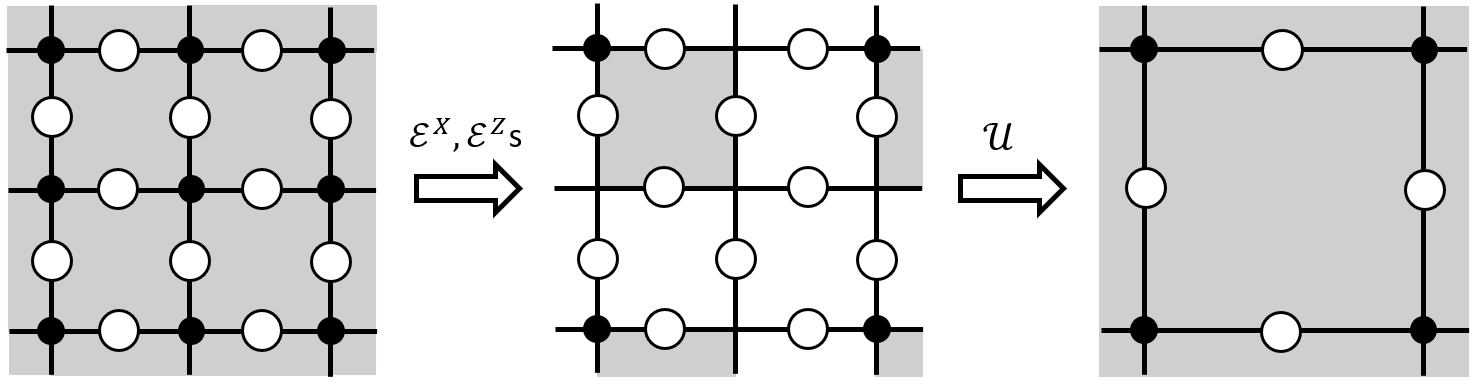}
    \caption{
    \textbf{RG scheme for the thermal toric code state--} In all panels, a plaquette (vertex) is shaded (dotted) if has a non-zero probability of holding an $m$- ($e$-) anyon, and physical qubits are associated with edges of the lattice and drawn as circles. (left$\rightarrow$mid) $\mathcal{E}^X$ and $\mathcal{E}^Z$ (see Eq.\eqref{eq: theram_tc_E_X} and Eq.\eqref{eq: theram_tc_E_Z}) acts on each $2\times2$ block of plaquettes and vertices, respectively. The resulting state has anyons on one of its sublattices' plaquettes and vertices. (mid$\rightarrow$right) After disentangling with the unitary $\mathcal{U}$ depicted in Eq.\eqref{eq: thermal_tc_disentangle} and discarding the decoupled qubits, the new state is still a toric code Gibbs state, but with renormalized temperature $p'$ (Eq.\eqref{eq: thermal_tc_p_iteration}) supported on a coarse-grained lattice.
    }
    \label{fig: thermal_tc}
\end{figure}

After applying $\mathcal{E}^{X(Z)}$ to each block of plaquettes (vertices), the resulting state only has anyons on plaquettes and vertices corresponding to a sublattice (see Fig.\ref{fig: thermal_tc}, middle panel). 
To complete one iteration of the RG, we need to discard some degrees of freedom and put the state on a coarse-grained lattice. This step can be achieved by a series of local unitary operators called elementary moves introduced in \cite{dennis2002topological, aguado2008entanglement}.

This step is most easily described graphically:
\begin{equation}\label{eq: thermal_tc_disentangle}
    \eqfig{4.2cm}{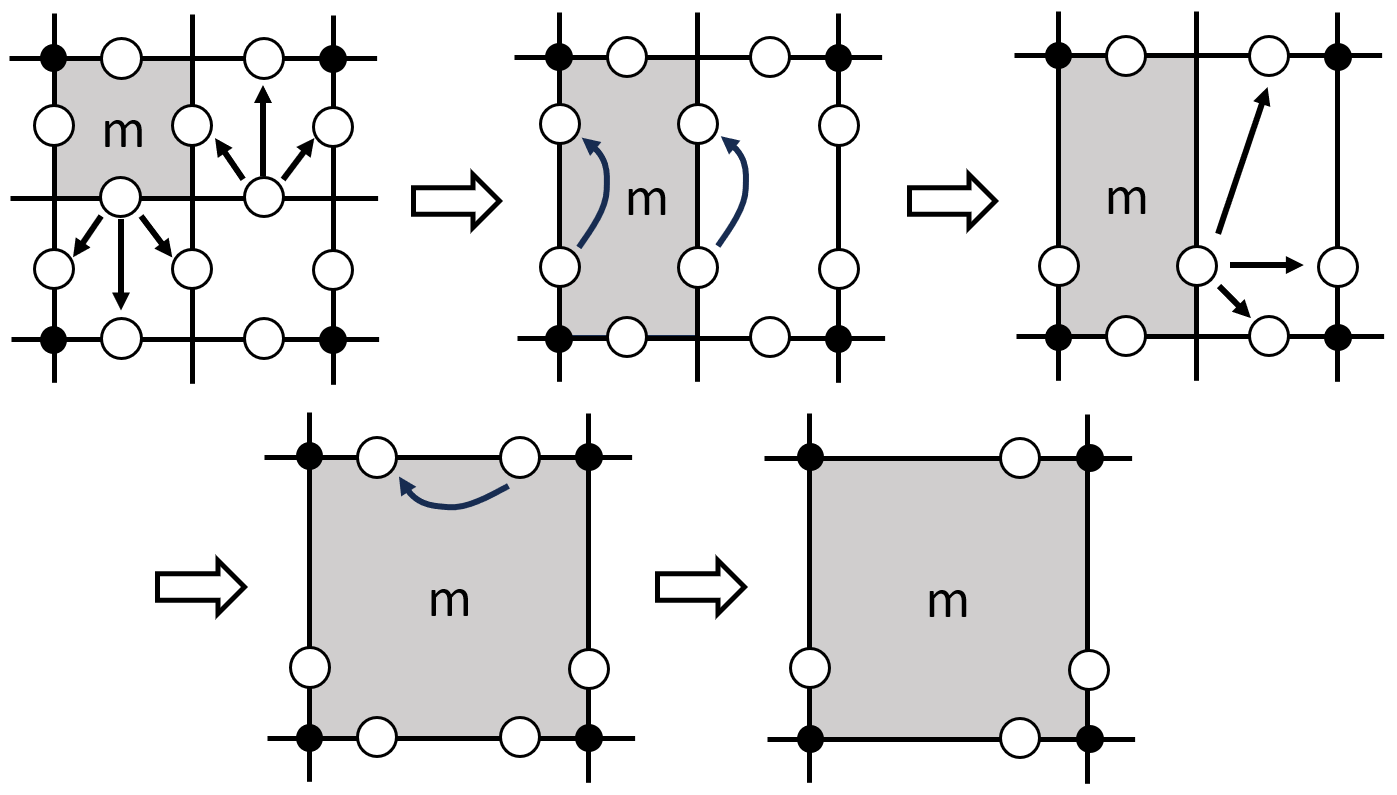}
\end{equation}
At each step, multiple controlled-not gates are applied, represented with arrows from the control qubit to target qubit. These gates decouple qubits into a product state which can then be removed, and the remaining qubits form a toric code state with anyons on a coarser lattice. Operations shown in each panel are applied in parallel to all the $2\times2$ blocks on the lattice. 


In summary, one iteration of the RG consists of:
\begin{equation}
\mathcal{C}
=
\mathcal{U} 
\circ 
\left(\bigotimes_{B\in\mathcal{B}'} \mathcal{E}_B^Z\right) 
\circ 
\left(\bigotimes_{B\in\mathcal{B}} \mathcal{E}_B^X\right) 
\end{equation}
where $\mathcal{B}$ contains $2\times2$ blocks of plaquettes and $\mathcal{B}'$ contains $2\times2$ blocks of vertices. $\mathcal{U}$ stands for the disentangling operations in Eq.\eqref{eq: thermal_tc_disentangle}. 

After one step of RG, a plaquette (vertex) contains an anyon if and only the four plaquettes (vertices) it was coarse-grained from contain an odd number of anyons.
The renormalized state is still a thermal toric code state, but with a renormalized probability (or renormalized temperature):
\begin{equation}\label{eq: thermal_tc_p_iteration}
\begin{aligned}
    & p'_{\beta'} = 4 p_\beta(1-p_\beta)^3 + 4 p^3_\beta(1-p_\beta)\\
    \Leftrightarrow\quad &\tanh{\beta'} = \tanh^4{\beta}
\end{aligned}
\end{equation}
We thus conclude that any finite temperature state $\rho_{\beta<\infty}$ flows to the infinite temperature one $\rho_{\beta=0}$ under the RG. 

Furthermore, since all channels in the RG are correlation-preserving with respect to their inputs, the RG is ideal and can be reversed. Thus there is the following bi-connection:
\begin{equation}
    \rho_{\beta=0} 
    ~\xrightarrow{{\rm RG}^{-1}}~
    \rho_{\beta} 
    ~\xrightarrow{\rm RG}~
    \rho_{\beta=0} \quad\quad \beta<\infty
\end{equation}
Since the infinite temperature state $\rho_{\beta=0}\propto \identity$ is in the trivial phase, we conclude that $\rho_\beta$ is also in the trivial phase. 

\section{Noisy toric code state} \label{sec: noisy_tc}

Now we consider the mixed-state obtained by applying noise to a pure toric code state. 
All notations in the section follow those introduced in Sec.\ref{sec: tc_review}.

The toric code model is naturally a quantum memory that stores quantum information in its ground state subspace $V$. 
In Sec.\ref{sec: the_thm}, we discuss the relation between the preservation of logical information and the preservation of the phase of matter. 
We prove that if an LC transformation $\mathcal{C}$ does not bring a pure toric code state out of its phase, then $\mathcal{C}$ must preserve any quantum information stored in $V$.

In Sec.\ref{sec: noisy_tc_rg} and \ref{sec: tMWPM}, we describe two ways to show that the dephased toric code state is in the toric code phase when dephasing strength is small. 
More specifically, we show that there exist LC transformations that bring the dephased state back to a pure toric code state. 
In Sec.\ref{sec: noisy_tc_rg}, the LC transformation is motivated by the Harrington decoder of the toric code and takes the form of an RG transformation. 
In Sec.\ref{sec: tMWPM}, the LC transformation is obtained by spatially truncating the minimum weight perfect matching (MWPM) decoder.

\subsection{Logical information and long-range entanglement} \label{sec: the_thm}
The toric code model, as its name suggests, is naturally a quantum error correcting code whose codespace is the ground state subspace $V$ (Eq.\eqref{eq: tc_V}). In this context, an important question is whether a noise channel $\mathcal{N}$ destroys logical information stored in a quantum memory. 
Mathematically, this is equivalent to asking whether there exists a recovery channel $\mathcal{R}$ such that:
\begin{equation}
    \mathcal{R}\circ\mathcal{N}(\ket{\psi}\bra{\psi}) = \ket{\psi}\bra{\psi} \ \ \ \forall \ket{\psi}\in V.
\end{equation}
In quantum error correction, $\mathcal{R}$ is often realized by a \textit{decoder}, which maps any input state into an output supported within $V$
\footnote{Since we are only concerned about in-principle recoverability of logical information, we assume all operations within $\mathcal{R}$ are noiseless.}.  If such $\mathcal{R}$ exists, we say the logical information is preserved by $\mathcal{N}$. Otherwise, we say the logical information is destroyed. 

To relate the phase of the mixed state, as defined by two-way LC connection, to preservation of logical information, we will need to first prove the following theorem.
\begin{theorem}\label{thm: thm2}
Let $V$ be the code subspace of a toric code defined on a torus. Suppose $\mathcal{C}$ is a local channel transformation satisfying $\supp ~\mathcal{C}(\ket{\psi}\bra{\psi})\subseteq V\ \ \forall \ket{\psi}\in V$, then $\mathcal C$'s action when restricted to $V$ is a unitary channel. 
\end{theorem}

To gain some intuition for why locality of the channel is essential in the theorem, consider the following channel:
\begin{equation}
    \mathcal{N}(\rho) := \frac{1}{2}\rho + \frac{1}{2}\tilde X_1 \rho \tilde X_1
\end{equation}
where $\tilde X_1=\Pi_{i\in S_1} X_i$ is the logical $X$ operator of the first encoded qubit (see Sec.\ref{sec: tc_review}). $\mathcal{N}$ is not an LC transformation: $\mathcal{N}(\ket{0}^{\otimes L}\bra{0}^{\otimes L})=\frac{1}{2}(\ket{0}^{\otimes L}\bra{0}^{\otimes L} + \ket{1}^{\otimes L}\bra{1}^{\otimes L})$,  a non-trivial mixed-state with long-range correlations.  Furthermore, $\mathcal{N}$ preserves $V$, but its action within $V$ is dephasing the first logical qubit, which is not a unitary action. 

We now prove Thm.\ref{thm: thm2}.
\begin{proof}
$\mathcal{C}$  can be dilated into an LU circuit $U$ that acts jointly on the physical qubits (referred to as $P$) and the ancilla qubits (referred to as $A$). Consider $U$'s action on a codeword state:
\begin{equation}
    \ket{\psi;\textbf{0}}:=\ket{\psi}_P\ket{\textbf{0}}_A 
    ~\xrightarrow{U}~
    \ket{\phi}_{PA}
\end{equation}
where $\ket{\psi}$ is any code word state in $V$. For later convenience we define the expanded codespace $V_0$, which is the subspace of $\hilbert_{PA}$ spanned by $\{\ket{\psi}_P\ket{\textbf{0}}_A\}_{\ket{\psi}\in V}$. We use $V_0$ to refer to both the subspace and the code defined by it. $V_0$ is still a stabilizer code, whose stabilizers are those of $V$ combined with $\{Z_i:\ i\in A\}$. 

\begin{figure}[h!]
    \centering
    \includegraphics[width=0.25\linewidth]{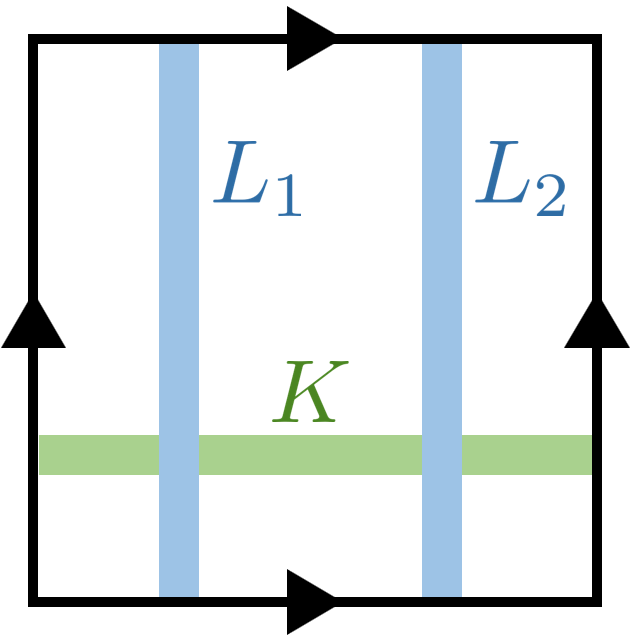}
    \caption{Geometry of operators $L_1$, $L_2$, and $K$.}
    \label{fig: in_proof}
    \end{figure}
Let $L_1$ and $L_2$ be two Pauli logical operators of the toric code that act in the same way in the code subspace $V$ (see Fig.\ref{fig: in_proof}). Thus $L_1 L_2$ is a stabilizer of the toric code. Since $\mathcal{C}$ preserves the code subspace, we have:
\begin{equation}
    \bra{\phi}L_1 L_2\ket{\phi} = \tr(C(\ket{\psi}\bra{\psi})L_1 L_2)=1
\end{equation}
This leads to:
\begin{equation}
    \bra{\psi;\textbf{0}}L_1^U L_2^U\ket{\psi;\textbf{0}} = 1
\end{equation}
where $L_i^U := U^\dagger L_i U$ has support on both $P$ and $A$, and is not necessarily a Pauli operator. 
Recalling that $L_1^U L_2^U$ is a unitary operator and the above expression holds for any $\ket{\psi; \textbf{0}}$, we conclude that $L_1^U L_2^U$ acts as logical identity in the extended codespace $V_0$. 

To proceed, we assume the spatial separation between $L_1$ and $L_2$ to be much larger than the range of $U$, so that $L_1^U$ and $L_2^U$ are also well-separated. 

We claim that both $L_1^U$ and $L_2^U$ are logical operators of $V_0$. Otherwise, there needs to be a codeword state $\ket{a}\in V_0$ such that $L_1^U\ket{a}\notin V_0 $. This implies at least one stabilizer $S$ of $V_0$ is violated by the state: $\bra{a}(L_1^U)^\dagger S L_1^U\ket{a} \neq 1$, and $S$ must have spatial overlap with $L_1^U$. Further, since $L_2^U$ is far from both $L_1^U$ and $S$, 
\begin{equation}
     \bra{a}(L_1^U L_2^U)^\dagger S L_1^UL_2^U\ket{a} \neq 1
\end{equation}
But this cannot be true because $L_1^U L_2^U\ket{a}=\ket{a}$ and $S$ is a stabilizer. 

The same reasoning applies to any Pauli logical operator, referred to as $K$, whose spatial support is perpendicular to $L=L_1$ (see Fig.\ref{fig: in_proof}). By varying $K$ and $L$, their product $R=K\cdot L$ can represent all of the $15$ inequivalent Pauli logical operators of the toric code. We fix such a set: $\mathcal{P}=\{R_{1}, R_{2}, ..., R_{15}\}$. The image of $\mathcal{P}$ under $R\rightarrow R^U$ is a set of 15 logical operators of $V_0$, as we just proved. Furthermore, since the map $R\rightarrow R^U$ preserves all the multiplication and commutation relations, we know $\mathcal{P}^U$ must act as a set of 15 inequivalent Pauli logical operators on $V_0$, up to a basis rotation.

We consider the part of $R^{U}_{i}$ (or $R_{i}$) when restricted to the codespace $V_0$ (or $V$):
\begin{equation}
\begin{aligned}
    \Pi_{V_0}R^U_i &= \tilde{R^U_i}\otimes\ket{\textbf{0}}\bra{\textbf{0}}\\
    \Pi_V R_i &= \tilde{R_i}
\end{aligned}
\end{equation}
where $\tilde{R^U_i}$ and $\tilde{R_i}$ are operators acting within $V$ only. 
$\Pi_V$ is the projector to the subspace $V$ and $\Pi_{V_0}=\Pi_{V}\otimes\ket{\textbf{0}}\bra{\textbf{0}}$. As explained, both $\{\tilde{R_i}\}$ and $\{\tilde{R^U_i}\}$ realize the algebra of Pauli operators in the logical space.

We have:
\begin{equation}
    \begin{aligned}
        \mathcal{C}(\tilde{R^U_i})
        &=
        \tr_{A}(U (\tilde{R^U_i}\otimes \ket{\textbf{0}}\bra{\textbf{0}}) U^\dagger)\\
        &=
        \tr_{A}(U (R_i^U \Pi_{V_0}) U^\dagger)\\
        &= 
        \tr_{A}(R_i U\Pi_{V_0}U^\dagger)\\
        &= {R_i} \mathcal{C}(\Pi_V)\\
        &=
        \tilde{R_i} \mathcal{C}(\Pi_V)
    \end{aligned}
\end{equation}
The second to last equality holds because $R_i$ is supported on $P$ only, while the last one holds because $\supp~\mathcal{C}(\Pi_V)\subseteq V$ by assumption.

On the \textit{r.h.s.} of the second equality above $R_i^U$ and $\Pi_{V_0}$ commute. Thus if we change their order then the same derivation gives:
\begin{equation}
    \mathcal{C}(\tilde{R^U_i}) = \mathcal{C}(\Pi_V)\tilde{R_i}
\end{equation}
Since the relation holds for any $i\in\{1,...,15\}$ and $\supp~\mathcal{C}(\Pi_V) \subseteq V$, we know $\mathcal{C}(\Pi_V)\propto \Pi_V$. Further, since $\mathcal C$ is trace-preserving, we have $\mathcal{C}(\Pi_{V})=\Pi_V$. Thus:
\begin{equation}
    \mathcal{C}(\tilde{R^U_i}) = \tilde{R_i}
\end{equation}
This implies that when restricted to $V$,  $\mathcal{C}(\cdot)$ is a $*$-isomorphism and must be a unitary channel.
\end{proof}

We use the theorem to explore the relation between the phase of the noisy toric code state and the preservation of quantum information stored.

Consider a toric code's codeword state $\ket{\psi}\in V$. Suppose $\mathcal{N}_\psi$ is an LC transformation that preserves the toric code phase. By definition, there exists another LC transformation $\mathcal{D}_\psi$ such that $\supp~\mathcal{D}_\psi\circ\mathcal{N}_\psi(\ket{\psi}) \subseteq V$ \footnote{We emphasize that the condition should not be $\mathcal{D}_{\psi}\circ\mathcal{N}_\psi(\ket{\psi})=\ket{\psi}$, according to our definition of the toric code phase in Sec.\ref{sec: tc_review}}. We first point out that the pair $(\mathcal{N}_\psi, \mathcal{D}_\psi)$ satisfies
\begin{equation}\label{eq: tc_state_independence}
   \mathcal{D}_{\psi}\circ\mathcal{N}_\psi(\ket{\psi'}\bra{\psi'})\subseteq V~~~\forall \ket{\psi'}\in V,
\end{equation}
which we prove in App.\ref{app: tc_state_independence}. Since the choice of $(\mathcal{N}_\psi, \mathcal{D}_\psi)$ does not depend on the codeword state $\ket{\psi}$, we drop the $\psi$ subscript henceforth.

The channel $\mathcal{D}\circ\mathcal{N}$ thus satisfies the condition in Thm.\ref{thm: thm2}, according to which we have 
\begin{equation}
    \mathcal{D}\circ\mathcal{N}(\ket{\psi}\bra{\psi})=U\ket{\psi}\bra{\psi}U^\dagger~~~\forall \ket{\psi}\in V
\end{equation}
for some logical unitary operator $U$. 

We thus conclude that if the noise $\mathcal{N}$ preserves the toric code phase, then it also preserves the logical information stored. In particular, the recovery map can be chosen as $R(\cdot)=U^\dagger\mathcal{D}(\cdot)U$. 


We consider a more detailed scenario where the channel $\mathcal{N}=\mathcal{N}_p$ has a strength parameter $p$.
When the noise is very strong, both the toric code phase and the logical information stored should be destroyed. Thus one can define two critical noise strengths: $p_{\rm t.c.}$, beyond which the noisy state is no longer in the toric code phase; and $p_{\rm coding}$, beyond which the stored logical information is no longer recoverable. The previous analysis shows that
\begin{equation}
    p_{\rm t.c.} \leq p_{\rm coding}
\end{equation}
Namely, the loss of logical information must occur after transitioning out of the toric code phase. 

If there is a gap between $p_{\rm t.c.}$ and $p_{\rm coding}$, then the noisy state $\mathcal{N}_p(\ket{\psi}\bra{\psi})$ for $p\in(p_{\rm t.c.}, p_{\rm coding})$ is not in the toric code phase but still contains logical information. In this case, the corresponding recovery map $\mathcal{R}$ that recovers logical information must be non-LC. 



\subsection{RG of the dephased toric code state} \label{sec: noisy_tc_rg}
We illustrate these general results in a specific example, for which we construct explicit RG channels.  We consider a toric code ground state $\ket{\rm t.c.} \in V$ subject to phase-flip noise with strength $p$ ( Eq.\eqref{eq: X_dephase}):
\begin{equation}\label{eq: dephased_tc}
    \rho_p := (\mathcal{N}^Z_p)^{\otimes L} \left(\ket{\rm t.c.}\bra{\rm t.c.}\right)
\end{equation}
We ask whether the state is in the same phase as $\ket{\rm t.c.}$, when $p$ is small. 

\begin{figure*}
    \centering
    \subfloat[]{\includegraphics[width=0.18\linewidth]{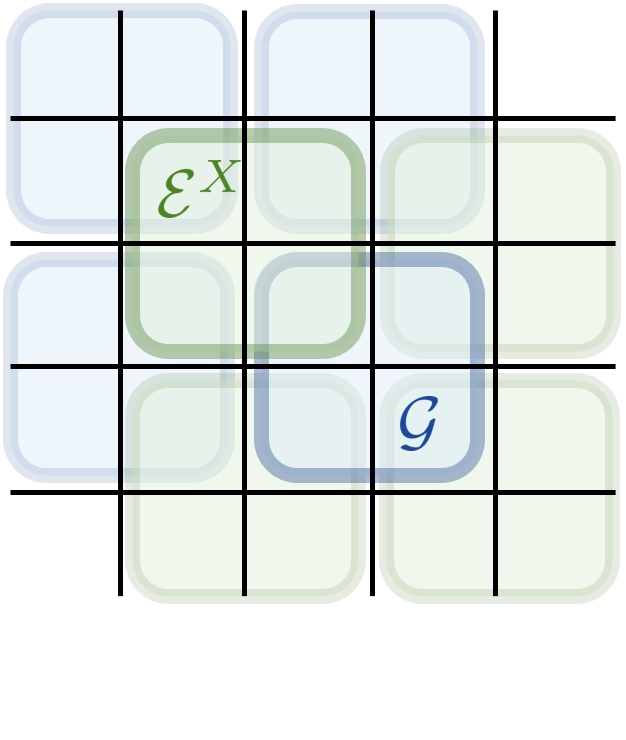}}
    \hspace{0.02\linewidth}
    \subfloat[]{\includegraphics[width=0.36\linewidth]{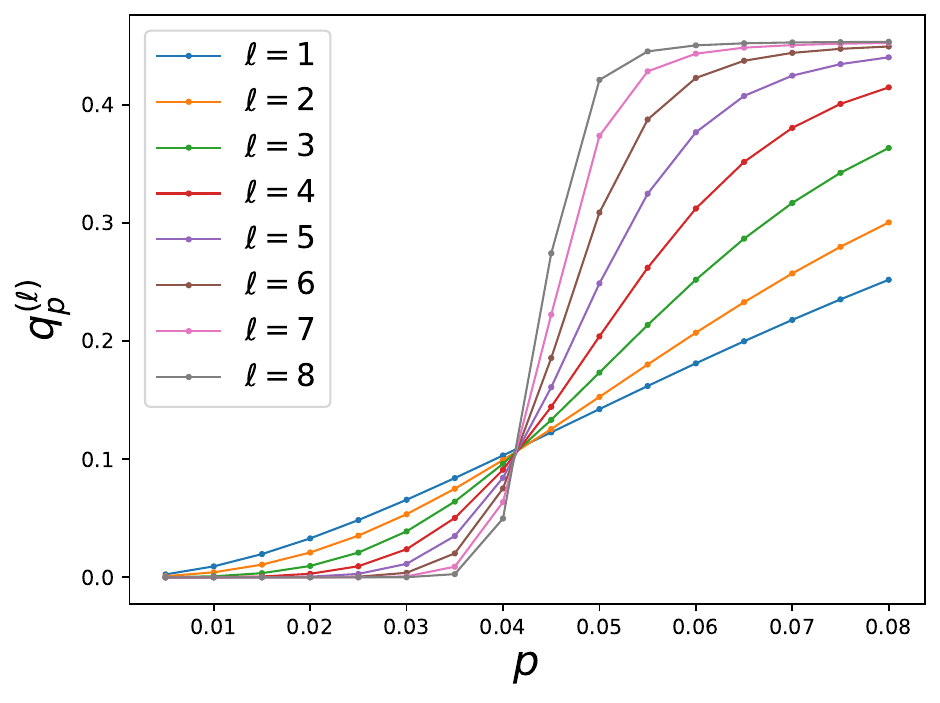}}
    \hspace{0.02\linewidth}
    \subfloat[]{\includegraphics[width=0.36\linewidth]{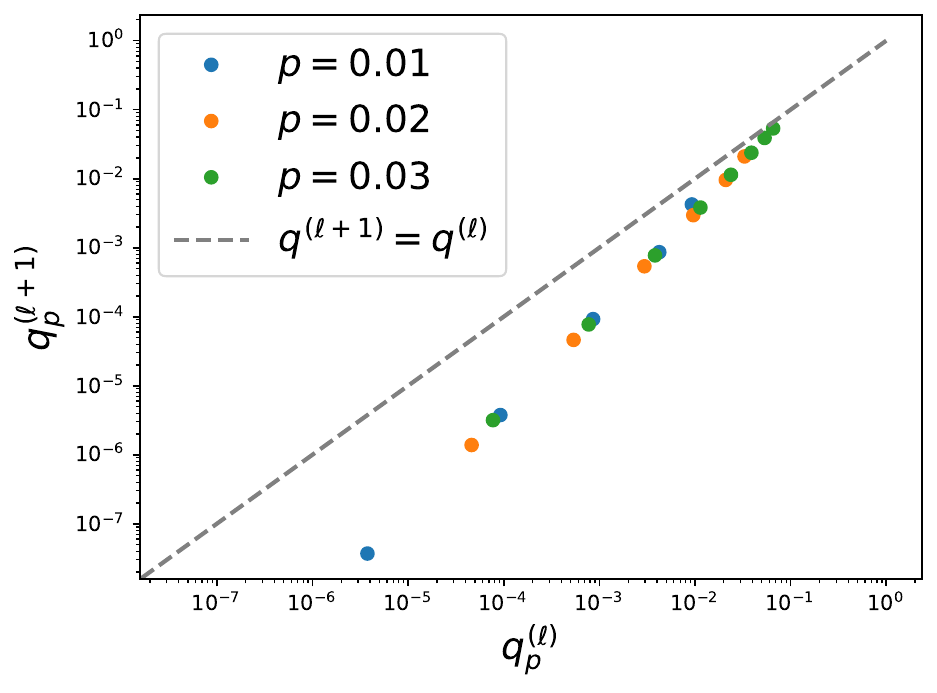}}
    \hspace{0.02\linewidth}
    \caption{
    \textbf{RG of the dephased toric code state--}
    (a) In each RG iteration, $\mathcal{G}$ is applied in parallel to all the odd blocks (blue), then $\mathcal{E}^X$ is applied to all the even blocks (green). Finally disentangling unitaries (see Eq.\eqref{eq: thermal_tc_disentangle}, not drawn in the figure) are applied to reduce the lattice size by half. 
    (b) RG flow of the anyon density $q_p^{(\ell)}$. 
    (c) Iteration relation of $q_p^{(\ell)}$ when approaching $0$, for various choices of $p$.
    }
    \label{fig:dephased_tc}
\end{figure*}

It is convenient to work in the anyon number basis Eq.\eqref{eq: anyon_basis}, and since all states in this example are in the $e$-anyon-free subspace, we omit the $\be$ labeling henceforth. An $Z$ operator acting on a qubit on an edge will create a pair of anyons in the two plaquettes adjacent to the edge. But if two anyons meet in the same plaquette, they annihilate. Thus if we fix the set of edges acted on by $Z$, then anyons appear on faces adjacent to an odd number of $Z$s (see Fig.\ref{fig: noisy_tc_anyon_config}). 

\begin{figure}[h!]
    \centering
    \includegraphics[width=0.35\linewidth]{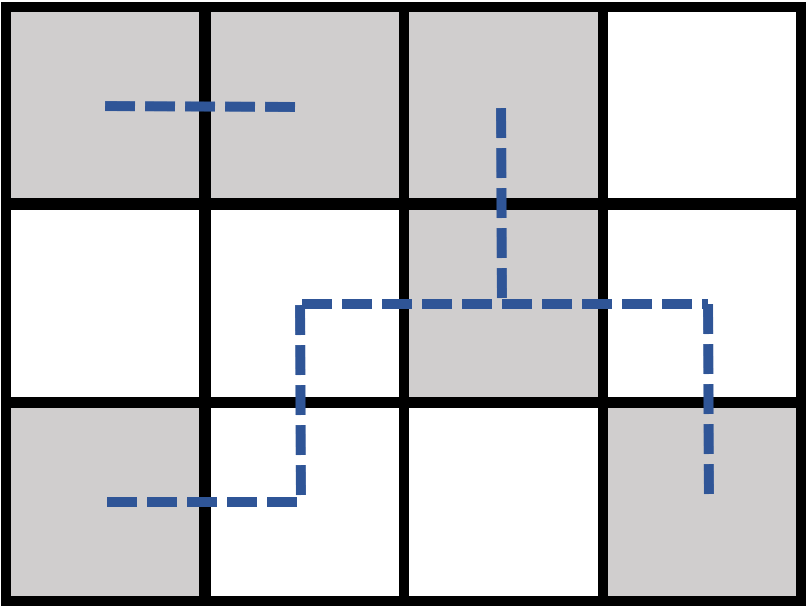}
    \caption{\textbf{A sample anyon configuration for noisy toric code}-- A dashed line on an edge denotes the corresponding qubit's phase is flipped by an $Z$ operator.  Anyons are created in plaquettes (shaded) where an odd number of dashed lines meet.}
    \label{fig: noisy_tc_anyon_config}
    \end{figure}
    

The noisy state is a classical mixture of anyon configurations which differ significantly from those of the Gibbs state. When $p$ is sufficiently small, the typical size of an error cluster is much smaller than the typical distance between clusters.  The errors create anyons at the boundary of each cluster. 



This picture suggests that by locally identifying clusters and pairing up anyons therein, one can remove all errors if $p$ is sufficiently small. 
This intuition underlies the design of several decoding algorithms for the toric code \cite{duclos2010fast,sergey2013self,breuckmann2016local,harrington2004analysis}, which aim to pair up anyons such that the quantum information stored in the code remains intact. 
As we show now, these decoders can be modified into RG schemes to reveal mixed-state phases of the noisy toric-code states.

We construct a simplified version of the Harrington decoder for the toric code \cite{harrington2004analysis, breuckmann2016local} to demonstrate that $\rho_p$ and $\ket{\rm t.c.}$ are in the same phase when $p$ is small.  We first partition the lattice into even blocks  $\mathcal{B}_{\rm even}$ and odd blocks $\mathcal{B}_{\rm odd}$ (see Fig.\ref{fig:dephased_tc}). Odd blocks are obtained by translating even blocks by one lattice spacing in both spatial directions. The two types of blocks will play different roles: coarse-graining will occur on even blocks and anyons will be paired up within odd blocks, regarded as boundary regions of even blocks. 

Each step of the RG is composed of three layers of local channels:
\begin{equation}
\mathcal{C}
=
\mathcal{U} 
\circ 
\left(\bigotimes_{B\in\mathcal{B}_{\rm even}} \mathcal{E}_B^X\right) 
\circ 
\left(\bigotimes_{B\in\mathcal{B}_{\rm odd}} \mathcal{G}_B\right)
\end{equation}
where the final step $\mathcal U$ is the disentangling operation depicted in Eq.\eqref{eq: thermal_tc_disentangle}, and $\mathcal{E}^X$ is the coarse-graining channel defined in Eq.\eqref{eq: theram_tc_E_X}. 

The main difference between this RG scheme and the one for the thermal toric code state (see Sec.\ref{sec: thermal_tc}) is the introduction of $\mathcal{G}$s. 
$\mathcal{G}_B$ annihilates all anyons within $B$ only if there are an even number of them; otherwise it leaves anyons within $B$ unmodified:
\begin{equation}
    \mathcal{G}_B := \sum_{\bm\in\{0,1\}^{\otimes 4}}\tilde{U}_{\bm}P_{\bm}(\cdot)P_{\bm}\tilde{U}_{\bm}^\dagger
\end{equation}
where $\tilde U_\bm$ equals $U_\bm$ (Eq.\eqref{eq: theram_tc_E_X}) if $\pi(\bm)=0$, and is $\identity$ when $\pi(\bm)=1$. 
The heuristic reason for introducing $\mathcal{G}_{B}$ is to pair up anyon clusters across the boundaries of the even blocks before the coarse-graining step. If such anyons were not paired, the coarse-graining on even blocks potentially prolongs them into clusters with a larger size, thus hindering effective anyon removal.

After each RG step $\mathcal{C}$, the new state is still an ensemble of different anyon configurations, albeit one that is not analytically tractable. Thus, we numerically compute how the RG steps affect the anyon density:
\begin{equation}
    q_p^{(\ell)} = \frac{1}{|P^{(\ell)}|}\sum_{\square\in P^{(\ell)}} \tr\left(\rho_p^{(\ell)}\frac{1- A_{\square}}{2}\right)
\end{equation}
where $P^{(\ell)}$ is the set of plaquettes on the renormalized lattice and $\rho_p^{(\ell)}$ is the renormalized state after $\ell$ iterations. $q=0$ implies the state is in the ground state subspace $V$ of the toric code Hamiltonian Eq.\eqref{eq: tc_hamiltonian}.

We use Monte Carlo method to study the flow of $q^{(\ell)}_p$ under RG. The simulation (Fig.\ref{fig:dephased_tc} (b)) shows that there is a sharp transition of $q_p^{(\ell)}$ at $p_c\approx 0.041$: 
\begin{equation}
    \lim_{l\rightarrow \infty} q_p^{(\ell)} =
    \left\{
    \begin{array}{lr}
        0 &  p<p_c\\
        O(1) & p>p_c
    \end{array}
    \right.
\end{equation}
When $p<p_c$ the RG successfully annihilates all anyons and the fixed-point state is in the ground state subspace of $V$; while when $p>p_c$, the fixed-point state has finite anyon density. 

Furthermore, when the anyon density $q^{(\ell)}$ approaches $0$, it transforms under each RG step as (see Fig.\ref{fig:dephased_tc} (c)):
\begin{equation}
    q^{(\ell+1)} \simeq (q^{(\ell)})^\gamma
\end{equation}
for $\gamma>1$. This behavior guarantees that a small number of iterations is sufficient for the convergence to the toric code ground state subspace (see the appendix App.\ref{app: convergence_noisy_tc} for a details).

We thus obtain the LC bi-connection:
\begin{equation}
    \ket{\rm t.c.} 
    ~\xrightarrow{{\rm noise}}~
    \rho_{p} 
    ~\xrightarrow{\rm RG}~
    \ket{\rm t.c.'} \quad\quad p<p_c
\end{equation}
where $\ket{\rm t.c.'}$ is another toric code state. 
This shows that $\rho_{p}$ is in the same phase as the pure toric code state, when $p<p_c$ and therefore $X$-dephasing noise is an irrelevant perturbation to the topologically ordered phase. 

We emphasize that this analysis does not show that $\rho_p$ with $p>p_c$ is in a different phase, because no bi-connection has been identified with this decoder. In fact, in the next section, we will construct another local channel that establishes that the phase boundary of the toric code phase extends to a much higher $p_c$.



\subsection{Truncated minimal weight perfect matching channel}\label{sec: tMWPM}
The seminal work \cite{dennis2002topological} showed that the dephased toric code state (Eq.\eqref{eq: dephased_tc}) retains its logical information up to a critical point $p_{\rm coding} \approx 0.108$, by relating the coding phase transition to the ferromagnetic-paramagnetic transition in the random bond Ising model.
A recovery channel called the maximal likelihood decoder \cite{dennis2002topological, bravyi2014efficient} decodes the logical information for any $p<p_{\rm coding}$, but the channel is not an LC transformation. 

The minimal weight perfect matching (MWPM) decoder is another decoder introduced in \cite{dennis2002topological}. It has a decoding threshold $p_{\rm MWPM}\approx 0.103$  very close to $p_{\rm coding}$ \cite{wang2003confinement}. 
The MWPM decoder, as a quantum channel, is also not an LC transformation. In the rest of this section, we show that it is possible to approximate the MPWM decoder's action arbitrarily well with an LC transformation whenever $p<p_{\rm MWPM}$. Consequently, we show that any dephased toric code state with $p<p_{\rm MWPM}$ is in the toric code phase. 

The core component of the MWPM decoder (henceforth referred to as $\mathcal{C}^{\rm MWPM}$) is a classical algorithm that solves the MWPM problem, namely looking for an anyons pairing scheme that minimizes the total length of the strings connecting pairs. Afterward, the decoder annihilates each anyon pair by acting with the string of $X$ operators connecting the pair. 

\begin{figure}[h]
    \centering
    \subfloat[]{\includegraphics[width=0.3\linewidth]{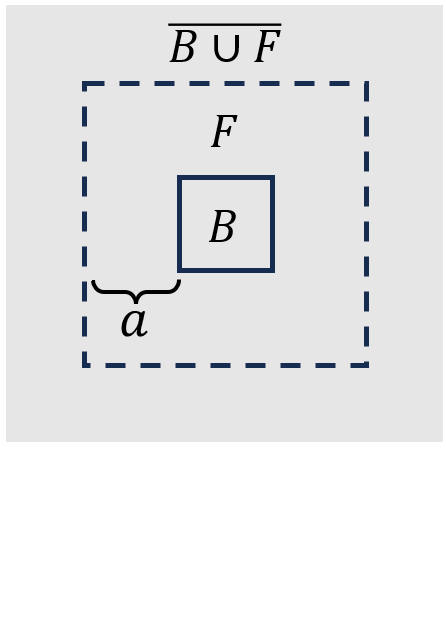}}
    \hspace{5pt}
    \subfloat[]{\includegraphics[width=0.65\linewidth]{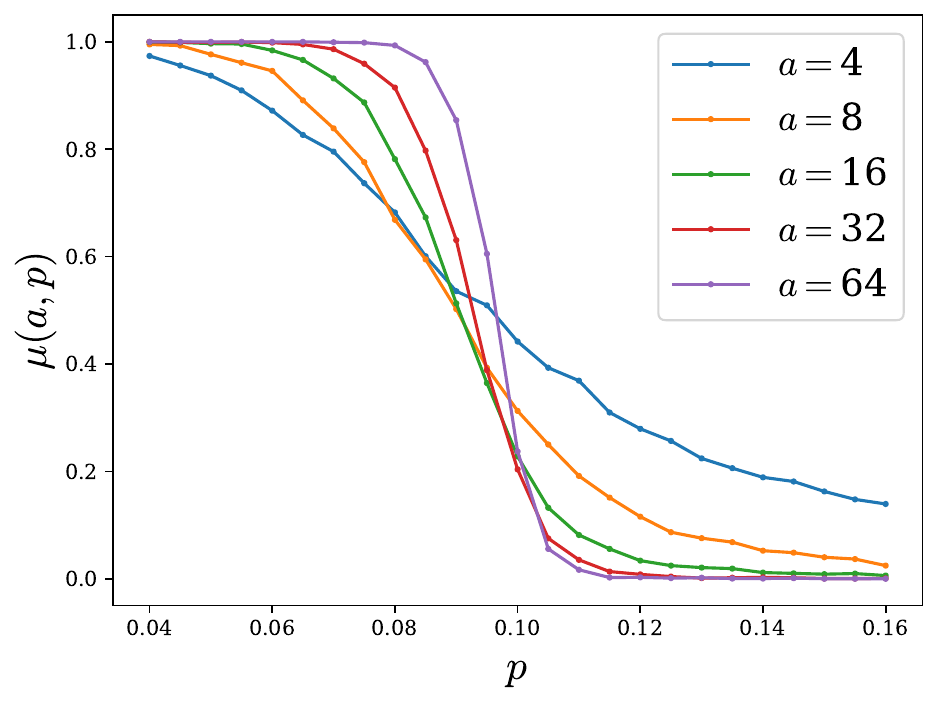}}
    \caption{\textbf{Truncated minimum weight perfect matching channel}-- (a) For a given block $B$, the corresponding channel $\mathcal{E}_{B, a}$ acts on both $B$ and a buffer region $F$ of a width $a$. (b) $\mu(a,p)$ is the probability that the truncated and global MWPM algorithms produce the same anyon pairings. It is plotted against $p$ for various choices of $a$. }
    \label{fig: mwpm}
\end{figure}

We now devise a way to truncate the $\mathcal{C}^{\rm MWPM}$ into an LC transformation. We first partition plaquettes into disjoint blocks, each with a size $b\times b$. 
For each block $B$, we apply a local channel $\mathcal{E}_{B, a}$ which acts jointly on $B$ and a buffer region $F$ of width $a$ surrounding $B$ (see Fig.\ref{fig: mwpm}). 
The local channel first solves the MWPM of anyons within the truncated region $B\cup F$, with the additional requirement that each anyon can either pair with another anyon or with the outer boundary of $F$ (the dashed line in Fig.\ref{fig: mwpm}). Then given the pairing scheme suggested by the MWPM solution, the channel only accepts a subset of it, namely pairs with at least one anyon within $B$. 
The truncated MWPM (tMWPM) channel applies the above channel to every block:
\begin{equation}
    \mathcal{C}^{\rm tMWPM}_a := \prod_{B\in \mathcal{B}} \mathcal{E}_{B, a}
\end{equation}
Note that different $\mathcal{E}_{B,a}$ can have overlapping domains. But since each $\mathcal{E}_{B, a}$ acts only on a patch of $(b+2a)^2$ qubits, we can always rearrange $\{\mathcal{E}_{B, a}\}$ into an $O((b+2a)^2)$-layer circuit so that each layer is composed of channels with non-overlapping domains.
After the rearrangement, it is apparent that $\mathcal{E}^{\rm tMWPM}_{a}$ is a range-$O((b+2a)^4)$ LC transformation (because both the depth and the range of gate is $O((b+2a)^2)$). 

The assumption behind the design of the tMWPM channel is the existence of a correlation length $\xi(p)$,
such that when $a\gg\xi(p)$, solving MWPM on $B\cup F$ only and solving MWPM on the whole system produce the same pairing for anyons in $B$. 
If the assumption holds for every block $B$, then the $\mathcal{C}_a^{\rm tMWPM}$ pairs all anyons in the same way as $\mathcal{C}^{\rm MWPM}$:
\begin{equation}
    \mathcal{C}^{\rm tMWPM}_{a}(\rho_p) \approx \mathcal{C}^{\rm MWPM}(\rho_p) = \ket{\rm t.c.} \quad a\gg\xi(p)
\end{equation}

We provide a rough estimate of how large $a$ needs to be for the `$\approx$' above to hold (agreement between local and global MWPM with probability $1-\epsilon$).  Given the correlation length assumption, the probability that for a single block $B$ the global and the truncated MWPM agrees should be $(1-e^{-a/\xi(p)})$. 
Assuming these probabilities for different blocks are independent (which should hold for far apart blocks), then we need
\begin{equation}
    (1-e^{-a/\xi(p)})^{L^2/b^2} > 1-\epsilon
\end{equation}
which occurs when
\begin{equation}
    a = \xi(p) O\left(\log\frac{L^2}{\epsilon}\right).
\end{equation}
$a$ diverges whenever $\xi(p)$ does, and this is expected to happen when $p\rightarrow p_{\rm MWPM}$.

To numerically support the assumption that there exists a correlation length for MWPM, 
we sample anyon configurations in the $X$-dephased toric code state and solve the MWPM first for all the anyons, and then only for anyons within $B\cup F$ (Fig.\ref{fig: mwpm}(a)). 
Then we compute the probability $\mu(a, p)$ that the two solutions are identical on $B$. 
We let both the system size $L$ and the diameter of $B$ be proportional to $a$, the width of the buffer $F$, so that the system has only one length scale $a$.

The simulation result is shown in Fig.\ref{fig: mwpm}(b) and suggests there is a critical point $p_{\rm tMWPM}$ in the interval $(0.10, 0.11)$, presumably consistent with $p_{\rm MWPM}$ in the thermodynamic limit. 
Below $p_{\rm tMWPM}$, we observe that $\lim_{a\rightarrow \infty}\mu(a, p)=1$. This indicates that the MWPM solution within $B$ is independent of anyons that are more than $O(\xi)$ away from $B$, for some correlation length $\xi$ which diverges at $p_{\rm tMWPM}$.  tMWPM thus serves as a local channel which, along with the noise channel, establishes the two-way connection demonstrating the toric code phase up to $p_{\rm tMWPM} \approx 0.1$.    Above $p_{\rm tMWPM}$, $\lim_{a\rightarrow \infty}\mu(a, p)=0$, implying non-locality in the MPWM solution. 

We point out that the simulation method above provides a way to detect the toric code phase using the anyon distribution data. One can fix a region $B$, implement MWPM on $B\cup F$, and gradually increase the buffer width $a$. If the MWPM solution when restricted to $B$ becomes stationary after $a$ is larger than some $a^*=O(1)$ with high probability, then the original mixed state is in the toric code phase because the tMWPM channel with $a\gtrsim a^*\log L$ can transform it into a pure toric code state. The method can potentially be used to detect mixed-state topological order in experiments.

\section{Discussion and outlook}
Our work provides two routes (RG and local versions of decoders) for constructing local channels connecting two mixed states to prove they are in the same phase.  We formulated a real-space RG scheme for mixed states and proposed the correlation-preserving property as a guiding criterion for finding coarse-graining maps; this property is necessary and sufficient for the map's action to be reversible (Thm.\ref{thm:equivalence}). 
We applied this formalism to identify the phases of several classes of mixed states obtained by perturbing a long-range entangled pure state with noise or finite temperature, and in particular we constructed an exact RG flow of the finite temperature 2D toric code state to infinite temperature.   

For toric code subject to decoherence, we also established a relation between the mixed state phase of the toric code and the integrity of logical information. In Thm.\ref{thm: thm2}, we proved that if local noise preserves the long-range entanglement of the toric code (and the resulting mixed state remains within the same phase as toric code), it must also preserve logical information encoded in the initial pure state. We conjecture that the converse statement is also true, namely, if local noise destroys the long-range entanglement of toric code, it must also destroy any encoded logical information. Even though the theorem and subsequent discussion focused on the toric code state, the main proof idea generalizes to many other topological codes and their corresponding phases. 

\begin{itemize}
    \item After formalizing the definition of mixed-state phase, one natural question to ask is whether there is a nontrivial phase that has no pure state nor classical state in this phase. A promising candidate is the ZX-dephased toric code state recently studied in \cite{wang2023intrinsic}. Since the state (when noise is strong) loses logical information \cite{wang2023intrinsic}, it is provably not in the toric code phase according to our Thm.\ref{thm: thm2}. Thus if the state is not in the trivial phase, it is an example of intrinsic mixed-state topological order. Another class of potential examples are decohered critical ground states \cite{zou2023channeling, lee2023quantum}, because such states naturally sit between a long-range entangled pure state and a long-range correlated classical state.
    \item One related question is whether one can find a computable quantity to detect nontrivial mixed-state phases.
    Topological entanglement negativity has been successfully used as a probe of mixed-state topological order \cite{neg,fan2023diagnostics}, but its robustness under LC transformations needs to be studied further.
    \item The decoherence-induced toric code transition can also be understood as a separability transition of the mixed-state \cite{chen2023separability}, and it would be valuable to relate this perspective with the mixed state phase and local channel perspective. 
    \item Pure state RG methods like DMRG serve as powerful computational methods for analyzing many-body systems. It is thus important to develop a numerical implementation of our mixed-state RG scheme. To facilitate simulations, one needs to first find an efficient representation of the mixed state (\textit{e.g.} using tensor networks), then update it iteratively using exact or approximately correlation-preserving maps, obtained by solving the optimization problem Eq.\eqref{eq: mixed_optimization}. We leave this problem for future exploration.
    \item { Steady states of dissipative dynamics form an important class of mixed states, for which there are several recent proposals of defining stable phase of matter \cite{rakovszky2023defining, liu2024dissipative}. When dealing with such problems, usually the dynamics (\textit{e.g.} Lindbladian superoperator) is known while a description of its steady state is unknown; therefore techniques developed in the current work do not immediately apply. It is desirable to develop RG schemes that can operate on Lindbladians directly.}
    \item As presented in Sec.\ref{sec: tMWPM}, tMWPM also serves as a practical probe of the mixed state toric code phase using anyon measurements. However, in experiments imperfect measurements lead to a finite density of `fake' anyons as well as unprobed anyons. To address this, one needs to consider a specific model of measurement errors and perform more than one round of measurements. Another potential direction is generalizing tMWPM to other topologically ordered mixed-state phases in two or higher dimensions.
\end{itemize}

\begin{acknowledgements}
We thank Arpit Dua, Tyler Ellison, Matthew P.A. Fisher, Tarun Grover, Ethan Lake, Yaodong Li, Zi-Wen Liu, Tsung-Cheng Lu, Ruochen Ma, and Michael Vasmer for helpful discussions and feedback. We also thank Roger G. Melko and Digital Research Alliance of Canada for computational resources. SS acknowledges the KITP graduate fellow program, during which part of this work was completed.
This work was supported by the Perimeter Institute for Theoretical Physics (PI) and the Natural Sciences and Engineering Research Council of Canada (NSERC).  Research at PI is supported in part by the Government of Canada through the Department of Innovation, Science and Economic Development Canada and by the Province of Ontario through the Ministry of Colleges and Universities. 
This research was also supported in part by the National Science Foundation under Grant No. NSF PHY-1748958 and NSF PHY-2309135, the Heising-Simons Foundation, and the Simons Foundation (216179, LB)
\end{acknowledgements}

\bibliography{main.bib}

\appendix
\onecolumngrid
\section{Various short proofs}

\subsection{Relation between pure- and mixed- state phase equivalence} \label{app: mixed_and_pure}
We sketch a proof of the following statement: On a given lattice, two pure states $\ket{\psi_1}$ and $\ket{\psi_2}$ are of the same mixed-state phase if and only if there exists an invertible state $\ket{a}$ on the same lattice such that $\ket{\psi_1}$ and $\ket{\psi_2}\otimes \ket{a}$ are of the same pure state phase. A many-body ground state $\ket{a}$ is called an invertible state if there exists another state $\ket{\tilde a}$ such that $\ket{a}\otimes\ket{\tilde a}$ can be LU transformed into a product state.

Since a LU transformation is also a LC transformation, pure-state phase equivalence trivially implies mixed-state phase equivalence. 

Now we show the other direction. Assume that there exists a pair of LC transformations $\mathcal C_{1,2}$ such that $\mathcal C_1(\rho_1)\approx \rho_2$ and $\mathcal C_2(\rho_2)\approx \rho_1$. 
We use $U_1, U_2$ to denote the unitary circuit within the definition of $\mathcal C_1$ and $\mathcal C_2$.

We observe that: Since $C_1(\rho_1)\approx\rho_2$ and $\rho_{2}$ is a pure state, the state right before the tracing-out operation must factorize as:
\begin{equation}
    U (\ket{\psi_{1}}\otimes\ket{\textbf{0}})\approx \ket{\psi_{2}}\otimes \ket{a_1}
\end{equation}
where $\ket{a_1}$ is the state that supports on qubits to be traced out and is defined on the same lattice. Recalling the definition of pure-state phases, we conclude the (pure-state) phase equivalence: $\ket{\psi_1}\underset{{\rm pure}}{\sim} \ket{\psi_2}\otimes \ket{a_1}$. Following a similar argument, we can also obtain $\ket{\psi_{2}}\underset{{\rm pure}}{\sim} \ket{\psi_1}\otimes\ket{a_2}$. We thus have:
\begin{equation}
    \ket{\psi_1}
    \underset{{\rm pure}}{\sim} 
    \ket{\psi_2}\otimes \ket{a_1} 
    \underset{{\rm pure}}{\sim} 
    \ket{\psi_1}\otimes \ket{a_2}\otimes \ket{a_1}
\end{equation}
As a result, the state $\ket{a_1}\otimes\ket{a_2}$ is in the trivial phase, and $\ket{a_1}$ is either in the trivial phase or an invertible state. 
Note that the proof assumes there is no `catalyst' effect in phase equivalence relation: if $\ket{\psi_1}$ and $\ket{\psi_2}$ cannot be LU connected to each other, then neither does the pair $\ket{\psi_1}\otimes\ket{a}$ and $\ket{\psi_2}\otimes\ket{a}$, for any state $\ket{a}$.

\subsection{Derivation of Eq.\eqref{eq: approx_correlation_preserving2}} \label{app: approx_correlation_preserving}
Let us assume $\mathcal{E}_{A\rightarrow A'}$ approximately preserves the correlation of a bi-partite state $\rho_{AB}$:
\begin{equation}
    I_{A:B}(\rho)-I_{A':B}(\mathcal{E}_{A\rightarrow A'}(\rho))=\epsilon,
\end{equation}
After defining $W$ and $\sigma_{A'EB}$ the same way as in the proof of Thm.\ref{thm:equivalence}, the condition above is equivalent to:
\begin{equation}
    I_{B:E|A'}(\sigma_{A'EB})=\epsilon.
\end{equation} 
Using the result in \cite{fawzi2015quantum}, there exists a reconstruction map $\mathcal{T}_{A'\rightarrow A'E}$ to approximately reconstruct $\sigma_{A'EB}$ from $\sigma_{A'B}$, with the approximation error bounded as:
\begin{equation}
    \epsilon \geq -2\log_2 F(\sigma_{A'EB}, \mathcal{T}_{A'\rightarrow A'E}(\sigma_{A'B})).
\end{equation}

We still define the recovery map $\mathcal{D}_{A'\rightarrow A}$ using Eq.\eqref{eq:recovery_map}. We have:
\begin{equation}
\begin{aligned}
    &F(\rho_{AB}, \mathcal{D}\circ\mathcal{E}(\rho_{AB}))\\
    =&
    F\left(
    \tr_R\left(U_W^\dagger U_W\left((\rho_{AB})\otimes\ket{0}_R\bra{0}\right)U^\dagger_W U_W\right),
    \tr_R\left(U_{W}^\dagger\mathcal{T}(\sigma_{A'B})U_W\right)
    \right)\\
    \geq&
    F\left(
    U_W\left((\rho_{AB})\otimes\ket{0}_R\bra{0}\right)U^\dagger_W,
    \mathcal{T}(\sigma_{A'B})
    \right)\\
    =& F\left(\sigma_{A'EB}, \mathcal{T}(\sigma_{A'B})\right),
\end{aligned}
\end{equation}
where the inequality is due to the monotonicity of $F$ under quantum channels. Combining the two expressions above, we arrive at the approximate recoverability we want:
\begin{equation}
    \epsilon \geq -\log_2 F(\rho_{AB}, \mathcal{D}\circ\mathcal{E}(\rho_{AB}))
\end{equation}

\subsection{Derivation of Eq.\eqref{eq: operator_product}} \label{app: operator_product}
We assume that the two operators $o_1^{(\ell)}=o_1$ and $o_2^{(\ell)}=o_2$ have non-overlapping lightcones $\mathcal{L}(o_1)$ and $\mathcal{L}(o_2)$ in the circuit representation of $\mathcal{F}^{\dagger}$, as illustrated below:
\begin{equation*}
    \eqfig{3.5cm}{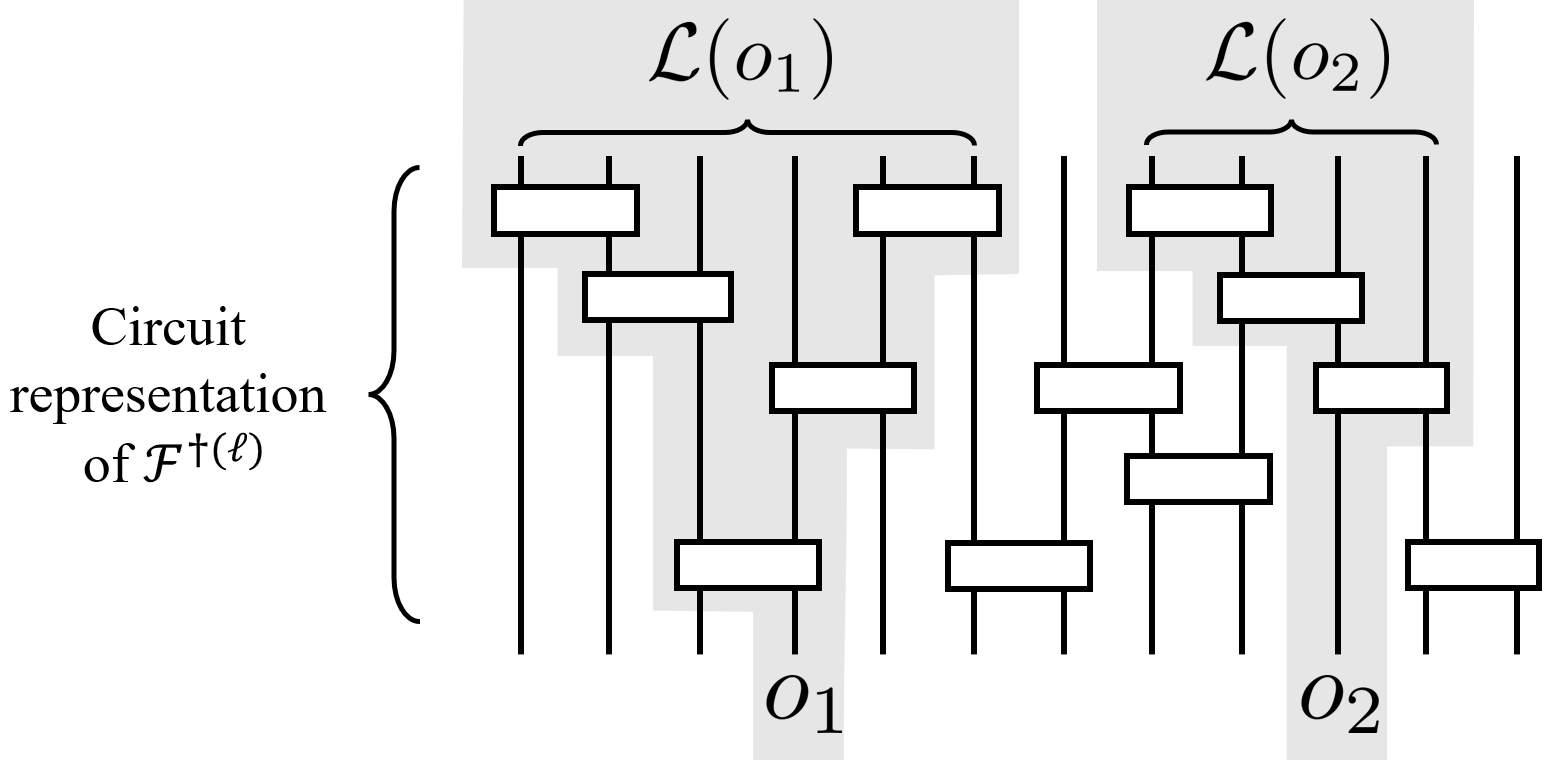}
\end{equation*}
To prove Eq.\eqref{eq: operator_product}, we decompose $\mathcal{F}^{\dagger}$ into $O(1)$ number of layers (as is drawn in the figure above). Each individual layer is of the form $\mathcal{E}^\dagger_1\otimes\mathcal{E}^\dagger_2\otimes...\otimes\mathcal{E}^\dagger_K$, where each $\mathcal{E}_{I}^\dagger$ is a dual channel that acts on a block of sites referred to as $B_I$ such that $B_I\cap B_J=\varnothing$ whenever $I\neq J$. We inspect $o_1 o_2$'s transformation under a single layer. Assuming $o_1$, $o_2$ are supported on $B_{I}$ and $B_J$ respectively, we have
\begin{equation} 
\mathcal{E}^\dagger_1\otimes\mathcal{E}^\dagger_2\otimes...\otimes\mathcal{E}^\dagger_K(o_1 o_2)= \mathcal{E}_I^\dagger(o_1)\mathcal{E}_J^\dagger(o_2)
\end{equation}
because dual channels are unital. By doing induction over all the layers in $\mathcal{F}^{\dagger}$, we arrive at our conclusion Eq.\eqref{eq: operator_product}.

\subsection{Derivation of Eq.\eqref{eq: ghz_dephase_channel_iteration}} \label{app: ghz_iteration}
We prove the channel equality by checking its action on $\ket{0}\bra{0}$, $\ket{0}\bra{1}$ and $\ket{1}\bra{1}$, which form a basis of linear operators for a single qubit.

We first act on the isometry and the noise channel:
\begin{equation}
    \begin{aligned}
        (\mathcal{N}_{p}^X)^{\otimes b} \circ \mathcal{U}_{w_b}(\ket{0}\bra{0}) 
        &= \sum_{\bs\in\{0,1\}^b} p^{|\bs|}(1-p)^{b-|\bs|}\ket{\bs}\bra{\bs} = \sigma_{00}\\
        (\mathcal{N}_{p}^X)^{\otimes b} \circ \mathcal{U}_{w_b}(\ket{0}\bra{1}) 
        &= \sum_{\bs\in\{0,1\}^b} p^{|\bs|}(1-p)^{b-|\bs|}\ket{\bs}\bra{\bar\bs}= \sigma_{01}\\
        (\mathcal{N}_{p}^X)^{\otimes b} \circ \mathcal{U}_{w_b}(\ket{1}\bra{1}) 
        &= \sum_{\bs\in\{0,1\}^b} p^{|\bs|}(1-p)^{b-|\bs|}\ket{\bar\bs}\bra{\bar\bs} = \sigma_{11}\\
    \end{aligned}
\end{equation}
Then we act the majority vote channel. Its action on $\sigma_{00}$ is:
\begin{equation}
\begin{aligned}
    \mathcal{E}(\sigma_{00})
    &=\sum_{\bs} p^{|\bs|}(1-p)^{b-|\bs|} \ket{\mathsf{maj}(\bs)}\bra{\mathsf{maj}(\bs)}\otimes\tr(\ket{\mathsf{diff}(\bs)}\bra{\mathsf{diff}(\bs)})\\
    &=\left(\sum_{\bs: |\bs|<b/2} p^{|\bs|}(1-p)^{b-|\bs|}\right)\ket{0}\bra{0} + \left(\sum_{\bs: |\bs|>b/2} p^{|\bs|}(1-p)^{b-|\bs|}\right)\ket{1}\bra{1}\\
    &= (1-p')\ket{0}\bra{0} + p' \ket{1}\bra{1}\\
    &= \mathcal{N}^X_{p'}(\ket{0}\bra{0}),
\end{aligned}
\end{equation}
where
\begin{equation}
    p' = \sum_{s:|s|>b/2} p^{|s|}(1-p)^{b-|s|} = \sum_{k=(b+1)/2}^{b} \binom{b}{k}p^{k}(1-p)^{b-k}
\end{equation}
Following a very similar calculation, we have:
\begin{equation}
\begin{aligned}
    \mathcal{E}(\sigma_{01}) 
    &= (1-p') \ket{0}\bra{1} + p' \ket{1}\bra{0}
    &=\mathcal{N}^X_{p'}(\ket{0}\bra{1})\\
    \mathcal{E}(\sigma_{11}) 
    &= (1-p') \ket{1}\bra{1} + p' \ket{0}\bra{0}
    &=\mathcal{N}^X_{p'}(\ket{1}\bra{1})
\end{aligned}
\end{equation}
Then we obtain Eq.\eqref{eq: ghz_dephase_channel_iteration}.

\subsection{Derivation of Eq.\eqref{eq: tc_state_independence}} \label{app: tc_state_independence}
The statement is equivalent to stating that $\mathcal{D}_\psi\circ\mathcal{N}_\psi(\ket{\phi})$ is anyon free, namely:
\begin{equation}
    \tr(\mathcal{D}_\psi\circ\mathcal{N}_\psi(\ket{\phi}) O)=1
    \quad
    \forall O\in\left\{A_\square, B_+\right\},~\ket{\phi}\in V
\end{equation}
which is further equivalent to:
\begin{equation}
    \tr(\ket{\phi}\bra{\phi}\mathcal{N}^\dagger_\psi\circ\mathcal{D}^\dagger_\psi(O))=1
\end{equation}
We use $A$ to denote the spatial support of  $\mathcal{N}^\dagger_\psi\circ\mathcal{D}^\dagger_\psi(O)$.
$A$ is a topologically trivial region because both $\mathcal N$ and $\mathcal{D}$ are LC transformations.
Thus the reduced density matrix of $\ket{\psi}$ and $\ket{\phi}$ over the region $A$ must be the same, and we can conclude the above equation holds for arbitrary $\ket{\phi}\in V$.

\section{Convergence of real-space RGs} \label{app: convergence}
In Sec.\ref{sec: rg_to_phase} it is stated that if a mixed state's RG flow $\{\rho^{(0)}, \rho^{(1)}, ...\}$ satisfies the following conditions for large enough $\ell$:
\begin{equation*}
\begin{aligned}
    &\text{Eq.\eqref{eq: convergence_condition1}:~~~} F(\rho^{(\ell)}, \rho^{(\infty)}) \simeq \exp(-\alpha\theta^{(\ell)} L^{(\ell)})\\
    &\text{Eq.\eqref{eq: convergence_condition2}:~~~} \theta^{(\ell+1)} \lesssim (\theta^{(\ell)}) ^ {\gamma} ~~~{\rm when}~~~ \theta^{(\ell)}\rightarrow 0_+ ,\\
\end{aligned}
\end{equation*}
then choosing
\begin{equation}
\ell^* = O(\log\log(L/\epsilon))
\end{equation}
guarantees $F(\rho^{(\ell^*)}, \rho^{(\infty)}) > 1-\epsilon$ for a small $\epsilon$. $L^{(\ell)} = L/b^{\ell}$ is the renormalized system size.

In this appendix we first prove the statement, then show the validness of conditions Eqs.\eqref{eq: convergence_condition1},\eqref{eq: convergence_condition2} in several scenarios, including all examples we studied in the main text. 

Suppose $\ell_0=O(1)$ simultaneously satisfies the following two conditions: Eq.\eqref{eq: convergence_condition2} holds when $\ell>\ell_0$ and that $\theta^{(\ell_0)}:=\theta_0<1$. Iterating the condition $(\ell-\ell_0)$ times, we get:
\begin{equation}
    \theta^{(\ell)} < \left(\theta_0 \right)^{\gamma^{\ell-\ell_0}}.
\end{equation}
We need to find how large $\ell$ needs to be, in order to satisfy:
\begin{equation}
    \exp(-\alpha\theta^{(\ell)} L^{(\ell)}) > 1-\epsilon,
\end{equation}
which is implied by (note that $L^{(\ell)}<L$):
\begin{equation}
    \theta^{(\ell)} < \epsilon/L,
\end{equation}
which is further implied by: 
\begin{equation}
    \ell > \ell_0 + \log_\gamma \log_{\theta_0^{-1}}(L/\epsilon) = O(\log\log(L/\epsilon))
\end{equation}
This completes the proof.

In the rest of this appendix section, we show several scenarios where the condition Eq.\eqref{eq: convergence_condition1}\eqref{eq: convergence_condition2} is satisfied. 
For the case of matrix product state with a tree tensor network RG, the analysis is done thoroughly in a recent work \cite{malz2023preparation}.

\subsection{Classical statistical mechanics models} \label{app: convergence_classical}
Let us consider a classical statistical mechanics model with the Hamiltonian:
\begin{equation}
    H_g = H_0 - g H'
\end{equation}
We assume that under some given RG process, $H_0$ is the RG fixed point while $H'$ is an irrelevant perturbation with respect to $H_0$. Further, we assume both $H_0$ and $H'$ are summations of spatially local terms each involving $O(1)$ number of sites. Further, each spin only appears $O(1)$ number of terms, and terms are uniformly bounded.

Let $\rho_g\propto \exp(-H_g)$ be the Gibbs state of $H_g$. We are interested in how fast it approaches the fixed-point $\rho_0$. We use the fidelity as a measure of closeness between the two states:
\begin{equation} \label{eq: fidelity_classical_states}
    F(\rho_g, \rho_0) = |\sqrt{\rho_g}\sqrt{\rho_0}|_1 = \tr(e^{-H_g/2}e^{-H_0/2}) / \sqrt{Z_g Z_0} = Z_{g/2} / \sqrt{Z_g Z_0}
\end{equation}
where $Z_g := \tr(\rho_g)$ is the partition function. We note that since $\rho_g$ is a classical state, the fidelity function $F$ coincides with the Bhattacharyya coefficient, a  measure of similarity for classical distributions.

We have:
\begin{equation}
    Z_g = \tr\left(e^{-H_0}\sum_{n}\frac{g^n}{n!} (H')^n\right) = Z_0 \sum_n \frac{g^n}{n!}\langle(H')^n\rangle_0
\end{equation}
Without loss of generality, we assume each term in $H'$ is positive, so that: $0\leq\langle(H')^n\rangle_0\leq|H'|^n$. Thus we have:
\begin{equation}
   Z_0  \leq Z_g \leq Z_0 e^{g |H'|}
\end{equation}
Combining with Eq.\eqref{eq: fidelity_classical_states}, we get:
\begin{equation}
    F(\rho_g, \rho_0) \geq e^{-g|H'|}
\end{equation}
Since $|H'|$ is a summation of spatially local terms whose norms are uniformly bounded, $|H'|$ must be upper-bounded by $\alpha L$ for some $\alpha=O(1)$ and large $L$. This leads to:
\begin{equation}
    F(\rho_g, \rho_0) \geq e^{-g\alpha L}
\end{equation}
Thus Eq.\eqref{eq: convergence_condition1} is satisfied. The Eq.\eqref{eq: convergence_condition2} is satisfied because $g$ is an irrelevant coupling for a non-critical fixed point.

\subsection{Noisy GHZ state in Sec.\ref{sec: noisy_ghz}} \label{app: convergence_noisy_ghz}

\noindent\textbf{(a) Bit flip noise}

The fidelity function is:
\begin{equation}
    \begin{aligned}
        F(\rho^X_{p, L},~ \rho^X_{0, L})
        &= \bra{\GHZ}\rho^X_{p, L}\ket{\GHZ}\\
        &= \frac{1}{2}\sum_{\bs\in\{0,1\}^{L}} p^{|\bs|}(1-p)^{L-|\bs|}\bra{\GHZ}(\ket{\bs}\bra{\bs}+\ket{\bar \bs}\bra{\bs}+\ket{\bs}\bra{\bar \bs}+\ket{\bar \bs}\bra{\bar \bs})\ket{\GHZ}\\
        &= (1-p)^{L} + p^{L}\\
        &\approx e^{-pL}
    \end{aligned}
\end{equation}
The approximation follows from $p\ll 1$.
The iteration relation of $p$ satisfies:
\begin{equation}
    p' =
   \sum_{k=(b+1)/2}^{b} \binom{b}{k}p^{k}(1-p)^{b-k}
    \leq 2^{b-1} p^{(b+1)/2}
    \leq p^{b/2}
\end{equation}
the last inequality holds if $p<2^{-2b+2}$, which can be always achieved after an $O(1)$ number of RG iterations starting from any $p\in (0, 0.5)$.

\noindent\textbf{(b) Phase-flip noise}

The fidelity function is:
\begin{equation}
    \begin{aligned}
        F(\rho^Z_{p, L},~ \rho^Z_{1/2, L})
        &= \tr\sqrt{\sqrt{\rho^Z_{1/2, L}}~\rho^Z_{p, L}~\sqrt{\rho^Z_{1/2, L}}}\\
        &= \frac{1}{2} + \frac{1}{2}\sqrt{1-(1-2p)^{2L}}
    \end{aligned}
\end{equation}
We notice that the fidelity function actually goes to $1$ when $L\rightarrow \infty$. This can be treated as a special case of the condition Eq.\eqref{eq: convergence_condition1} for $\alpha=0$.

The iteration relation for $p$ around the stable fixed-point $p=0.5$ is:
\begin{equation}
    (p'-1/2) = (p-1/2)^b
\end{equation}

\subsection{Thermal toric code state in Sec.\ref{sec: thermal_tc}}\label{app: convergence_thermal_tc}
We consider the fidelity between the finite temperature state and the infinite temperature one. Note that the latter is proportional to identity. Thus,
\begin{equation}
    \begin{aligned}
        F(\rho_{\beta},~ \rho_{0})
        &= 2^{-\frac{L^2}{2}}\tr\sqrt{\rho_{\beta}}\\
        &= 2^{-\frac{L^2}{2}}\sum_{\bm, \be, \bl}\sqrt{\Pr_m(\bm)\Pr_e(\be)\Pr_l(\bl)}\\
        &= 2^{-\frac{L^2}{2}+1}\left(\sum_{\bm}\sqrt{\Pr_m(\bm)}\right)^2\\
        &\geq 2^{-\frac{L^2}{2}+1}\left(\sum_{\bm:\pi(\bm)=0}p_\beta^{|\bm|/{2}}(1-p_\beta)^{(L^2/2-|\bm|)/2}\right)^2\\
        &= 2^{-\frac{L^2}{2}+1}(1-p_\beta)^{L^2/2}\left(\sum_{\bm:\pi(\bm)=0}e^{-\beta|\bm|}\right)^2\\
        &\geq 2^{-\frac{L^2}{2}+1}2^{-L^2/2} \left(2^{L/2-1} e^{-\beta L^2/2}\right)^2\\
        &= 2^{-\beta L^2 - 1}
    \end{aligned}
\end{equation}
The iteration relation of the inverse temperature $\beta$ is given by 
\begin{equation}
    \beta' = \tanh^{-1}\tanh^4 \beta \approx \beta^4
\end{equation}
at small $\beta$, which satisfies the condition.

\subsection{Noisy toric code state in Sec.\ref{sec: noisy_tc}} \label{app: convergence_noisy_tc}
The renormalized state's overlap with the ground state subspace $V$ is:
\begin{equation}
    F^{(\ell)}:= \tr\left(\rho^{(\ell)}\Pi_{V}\right)
\end{equation}
where $\Pi_{V}$ is the projection to $V$.

Since $\rho$ is always diagonal in the anyon number basis throughout the RG process, the quantity is the probability of having zero anyon after $l$ steps of RG. The quantity is bounded by the anyon density $q_p^{(\ell)}$:
\begin{equation}
    q_p^{(\ell)} = \sum_{a=1}^{L^2} \Pr(|m|=a) \frac{a}{L^2} \geq \frac{1}{L^2}\sum_{a=1}^{L^2} \Pr(|m|=a) = \frac{1}{L^2}(1-F^{(\ell)})
\end{equation}
leading to:
\begin{equation}
    F^{(\ell)}\geq 1-L^2 q_p^{(\ell)} \simeq e^{-L^2 q_p^{(\ell)}}.
\end{equation}
The last approximate equal sign holds only when $L^2 q_p^{(\ell)}$ is small. 

Furthermore, according to the numerical results Fig.\ref{fig:dephased_tc}(c), the decay of anyon density follows:
\begin{equation}
    q_p^{(\ell+1)} < (q_p^{(\ell)})^\gamma
\end{equation}
for some $\gamma>1$.

\section{RG of a mixed $\mathbb{Z}_2\times \mathbb{Z}_2$ SPT state} \label{app: spt_rg}

In this appendix, we consider the RG of a mixed SPT (symmetry protected topological) state. 
The problem was first considered in \cite{de2022symmetry}, where the authors use the string order parameter as a definition for the mixed state SPT. They show that, when a pure SPT state is subject to noise, the string order parameter (\textit{i.e.} the SPT phase) is preserved if and only if the noise is strongly symmetric, meaning that all its Kraus operators commute with the symmetry operator.

The above definition (via string order parameters) of mixed-state SPT actually agrees with the LC transformation based definition. To show this, we put forward a symmetric RG transformation that brings the noisy state back to a clean one. The RG is the same as the one proposed in \cite{lake2022exact}. Although in \cite{lake2022exact} the circuit is designed for recognizing pure-state phases, we point out that it can be readily applied on a mixed-state SPT states.

Here we demonstrate the principle with a simple (1+1)D SPT state, namely the $Z_2\times Z_2$ SPT. Consider a 1D lattice spin chain where in the bulk each site contains 2 qubits, labeled as $A$ and $B$. The pure SPT wavefunction can be written as:
\begin{equation}
    \ket{\psi} = \bigotimes_{i=-\infty}^{+\infty} \ket{\EPR_{i_B, (i+1)_A}},
\end{equation}
where each $\ket{\EPR_{i,j}}:=\frac{1}{\sqrt2}(\ket{0_i 0_j}+{\ket{1_i 1_j}})$. Note that the above way of defining the $Z_2\times Z_2$ SPT is related to the cluster state by rotating the two spins within a unit cell with a CNOT gate. It is straightforward to verify that the state has a $Z_2\times Z_2$ symmetry generated by 2 generators:
\begin{equation}
   U^X=\prod_i X_{i_A}X_{i_B}, \quad U^Z=\prod_i Z_{i_A}Z_{i_B}.
\end{equation}
Two generators define two string order parameters:
\begin{equation}
    S^X_{ij} = X_{i_B} \left(\prod_{k=i+1}^{j-1}X_{k_A}X_{k_B}\right) X_{j_A}, 
    \quad 
    S^Z_{ij} = Z_{i_B} \left(\prod_{k=i+1}^{j-1}Z_{k_A}Z_{k_B}\right) Z_{j_A}
\end{equation}
One signature of the SPT order in $\ket{\psi}$ is that expectation value of string order parameters does not decay with $|i-j|$.

We consider the $XX$ dephasing channel:
\begin{equation}
    \mathcal{N}^{XX}_p(\cdot) = (1-p)(\cdot) + p X_A X_B (\cdot)X_A X_B,
\end{equation}
whose action can be realized by the pair of spins within each site being flipped simultaneously with probability $p$.
The channel is strongly symmetric under $U_X$ and $U_Z$, because its Kraus operators commute with symmetry operators. 

The noisy SPT state of our interest is obtained by applying $\mathcal{N}_p$ uniformly upon $\psi$:
\begin{equation}
    \sigma_{p, L} = \mathcal{N}_p^{\otimes L}(\ket{\psi}\bra{\psi}).
\end{equation}

It will be useful to unravel $\sigma_p$ into a classical mixture of pure SPT states decorated by domain walls:
\begin{equation}
\begin{aligned}
    \sigma_p &= \E[\ket{\psi_{\bs}}\bra{\psi_{\bs}}]= \sum_{\bs\in\{0,1\}^L}P(\bs)\ket{\psi_{\bs}}\bra{\psi_{\bs}}\\
    P(\bs) &:= p^{|\bs|}(1-p)^{|\bs|}\\
    \ket{\psi_{\bs}} &:= \bigotimes_{i} X^{\bs_i}_{i_A} X^{\bs_{i+1}}_{(i+1)_B}\ket{\EPR_{i_B, (i+1)_A}}
\end{aligned}
\end{equation}

Now we focus on a block of $b$ sites within the state and devise the channel that coarse-grains the block. The part of the $\ket{\psi_\bs}$ within the block can be drawn as (for $b=3$):
\begin{equation*}
    \eqfig{2cm}{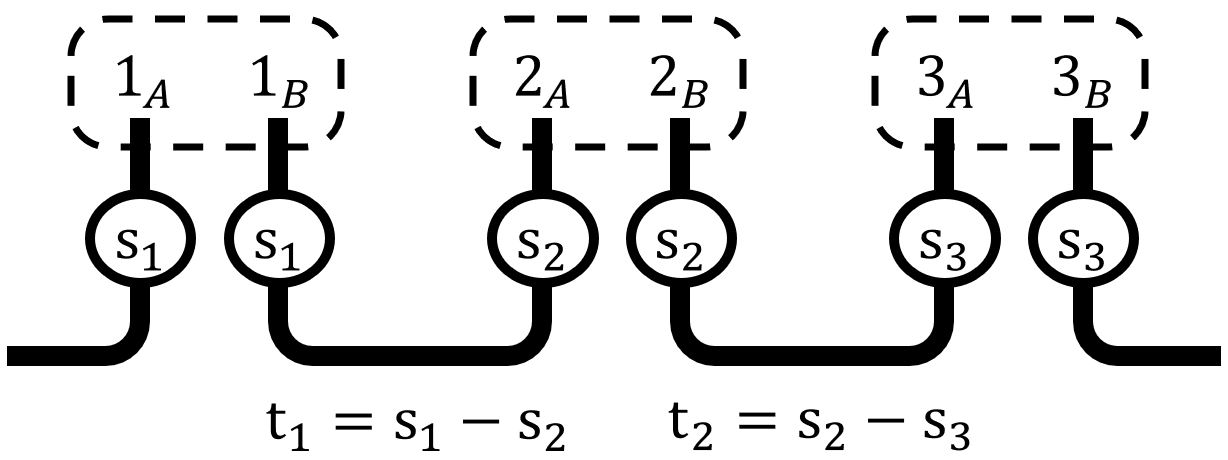}
\end{equation*}
where each hollowed circle is either an identity gate if $s=0$ or an $X$ gate if $s=1$. Each $s_i$ is an independent Bernoulli random variable with chance $p$. 

The renormalized site is formed by qubits $\{1_A, b_B\}$, while the $(b-2)$ entangled pairs supporting on the remaining qubits $\{1_B, 2_A, 2_B,..., b_A\}$ can help us infer and correct errors $s_1$ and $s_b$ before they get traced out. More concretely, the coarse-graining channel $\mathcal{E}$'s action is the following:
\begin{enumerate}
    \item Measure each Bell pairs \textit{within} the block in the $ZZ$ basis, whose outcome records the domain-walls between $s_i$, i.e., $\{t_1 = s_1-s_2, t_2 = s_3-s_2, ..., t_{b-1}=s_{b-1}-s_b\}$. The outcomes decide all the $s_i$ up to a global flip, \textit{i.e.} once we assume a value for the first error $s_1=\hat{s}_1$, the remaining $s$ are also uniquely determined by: $\hat{s}_k=(\sum_{j=1}^{k-1} t_j)-\hat{s}_1$.
    \item Assume 
    \begin{equation}
        \hat{s}_1=\argmax_{s_1 \in \{0,1\}} \Pr (s_1|t_1,...,t_{b-1})
    \end{equation}
    to be the actual error, and correct all sites within the block. Namely, applying $X_{i_A}X_{i_B}$ if $\hat s_i=1$. After this step, all Bell pairs within the block are noiseless and decoupled from $\{1_B, b_A\}$.
    \item Trace out $\{1_B, 2_A, 2_B,..., b_A\}$. Then we obtain the renormalized site with $\{1_B, b_A\}$,
\end{enumerate}

It is clear that channel $\mathcal{E}$ is strongly symmetric under $Z_2\times Z_2$ since each individual step is. The renormalized site gets a $XX$ error if and only if $\hat{s}_1\neq s_1$, whose probability we denote as $p'$. Thus we find that the normalized state is still a symmetrically dephased SPT state, but with a normalized noise strength $p'$:
\begin{equation}
    \mathcal{E}^{\otimes \frac{L}{b}}(\sigma^{\mathrm{SPT}}_{p,L}) = \sigma^{\mathrm{SPT}}_{p',L/b}
\end{equation}
The explicit form of $p'$ is:
\begin{equation}
    p' = \sum_{k = (b+1)/2}^{b} \binom{b}{k} p^{k} (1-p)^{b-k}
\end{equation}
Thus we obtain a similar RG flow as in the $X$-dephased GHZ state studied in Sec.\ref{sec: noisy_ghz}: symmetrically decohered SPT state flows back to a pure SPT state when $0\leq p < 0.5$.
\end{document}